%% file: main.tex
\documentclass[11pt]{article}
\usepackage[margin=1in]{geometry}
\usepackage{amsmath,amsthm, amsfonts, amssymb}
\usepackage{algorithm}
\usepackage[noend]{algpseudocode}
\usepackage{xcolor} 
\usepackage{parskip}
\usepackage{setspace}
\usepackage{xspace}
\usepackage{hyperref}
\usepackage{enumitem}
\usepackage{graphicx}
\usepackage{adjustbox}
\usepackage{caption}
\usepackage{mathtools}
\usepackage{thm-restate}
\usepackage{amsmath}
\usepackage{tikz}
\usepackage{float}
\usetikzlibrary{decorations.pathreplacing} 
\usetikzlibrary{arrows.meta,calc,positioning}
\usetikzlibrary{shapes.geometric}

\newtheorem{lemma}{Lemma}
\newtheorem{claim}{Claim}
\newtheorem{theorem}{Theorem}
\newtheorem{definition}{Definition}
\newtheorem{corollary}{Corollary}

\input{macros}

\title{Multi-Level Aggregation via Dual Fitting:\\ An $O(D)$-Competitive Algorithm}

\author{Sara Ahmadian\thanks{Google. \texttt{sahmadian@google.com}} \and Shuchi Chawla\thanks{University of Texas, Austin. \texttt{shuchi@cs.utexas.edu}} \and Ravi Kumar\thanks{Google. \texttt{ravi.k53@gmail.com}} \and Manish Purohit\thanks{Google. \texttt{purohitmanish89@gmail.com}} \and Shirley Zhang\thanks{Harvard University. \texttt{szhang2@g.harvard.edu }}}
\date{}

\begin{document}
\maketitle

\begin{abstract}
    We present a new online algorithm for the  well-known Multi-Level Aggregation Problem (\mlap) with arbitrary delay functions, achieving a $2D$-competitive ratio, where $D$ is the depth of the underlying tree. This result improves the current best-known competitive ratio of $O(D^2)$ and asymptotically matches the $D$-competitive bound previously known only for the deadline variant, thereby closing the asymptotic gap between the two settings. 

    Our key technical contribution is a novel dual fitting framework that provides a unified analysis for both settings; in particular, it also establishes a $D$-competitive ratio for \mlap with deadlines. Our analysis is built upon two new ideas: a \textit{hindsight dual construction}, which resolves the infeasibility issues in traditional online primal-dual methods, and a \textit{time-dependent dual packing} that maintains feasibility over dynamic request sets. 
\end{abstract}

\thispagestyle{empty}
\setcounter{tocdepth}{2}

\begin{footnotesize}
\tableofcontents
\end{footnotesize}

\thispagestyle{empty}
\setcounter{page}{0}

\newpage

\input{intro}

\input{prelim}

\input{deadlines_v2}

\input{delays_v4}
\input{deferred}
\input{extra_deadlines}
\input{tight_example}
\input{related}

\bibliographystyle{alphaurl}
\bibliography{bib}

\end{document}

%% file: macros.tex
\newenvironment{proofof}[1]{{\em Proof of #1.  }}{\hfill\qed}

\newcommand{\jrp}{\textsc{JRP}\xspace}
\newcommand{\mlap}{\textsc{MLAP}\xspace}
\newcommand{\mlapd}{\textsc{MLAP-deadline}\xspace}
\newcommand{\mlapdelay}{\textsc{MLAP-delay}\xspace}
\newcommand{\lp}{\textsc{LP}\xspace}

\renewcommand{\min}{\operatorname*{min}}
\renewcommand{\max}{\operatorname*{max}}

\newcommand{\UponCritical}{\textsc{UponCritical}}
\newcommand{\Explore}{\textsc{Explore}}
\newcommand{\Add}{\textsc{Add}}
\newcommand{\Invest}{\textsc{Invest}}
\newcommand{\Charge}{\textsc{Charge}}
\newcommand{\Simulate}{\textsc{Simulate}}
\newcommand{\LetState}[2]{\State \textbf{let} #1 $\leftarrow$ #2} 

\newcommand{\costalg}{\mathrm{ALG(\mc{T}, \mc{J})}}
\newcommand{\costopt}{\mathrm{OPT(\mc{T}, \mc{J})}}
\newcommand{\delayalg}{\mathrm{Delay(ALG)}}
\newcommand{\treecostalg}{\mathrm{TreeCost(ALG)}}
\newcommand{\prev}{\mathrm{prev}}
\newcommand{\nxt}{\mathrm{next}}
\newcommand{\cst}{\mathrm{cost}}

\newcommand{\children}{\operatorname{Children}}
\newcommand{\delay}{\operatorname{delay}}
\newcommand{\served}{\operatorname{served}}

\DeclareMathOperator*{\argmin}{arg\,min}

\newcommand{\del}{\operatorname{delay}}

\newcommand{\cost}{\ensuremath{w}}
\newcommand{\tree}{\mathcal{T}}

\newcommand{\sattime}{\mathrm{sattime}}
\newcommand{\reachtime}{\mathrm{reachtime}}
\newcommand{\reachset}{\mathrm{reachset}}
\newcommand{\freezetime}{\mathrm{freezetime}}
\newcommand{\patrons}{\mathrm{patrons}}
\newcommand{\paid}{\mathrm{paid}}
\newcommand{\mc}[1]{\mathcal{#1}}

\definecolor{reqblue}{RGB}{55,126,184}
\definecolor{reqgreen}{RGB}{77,175,74}
\definecolor{reqpurple}{RGB}{152,78,163}
\definecolor{reqorange}{RGB}{255,127,0}
\definecolor{reqred}{RGB}{228,26,28}
\definecolor{tickgold}{RGB}{253,191,111}

\tikzset{
  timeline/.style={line width=0.7pt, -{Stealth[length=3pt]}},
  reqwin/.style={rounded corners=1.6pt, line width=0.7pt},
  service/.style={line width=1pt},
  lbl/.style={anchor=east, font=\scriptsize},
  axlbl/.style={font=\scriptsize},
  treeedge/.style={line width=0.8pt},
  treenode/.style={circle, minimum size=6pt, inner sep=0pt, draw=black, line width=0.6pt},
}



%% file: intro.tex
\section{Introduction}

In modern interconnected systems, actions rarely occur in isolation but instead cascade through hierarchical structures---for instance, packages moving through multiple levels of local and regional hubs, or sensor data that is aggregated at intermediate nodes before transmission. A fundamental challenge in these settings is to balance the competing objectives of latency and cost. Providing immediate service minimizes delay but incurs high operational costs, whereas postponing requests cuts costs through aggregation (i.e., batching)  but at the expense of increased latency. Finding the right trade-off between these objectives is an important problem in domains such as logistics, supply-chain management, content delivery networks, and sensor data collection systems. 

The \emph{multi-level aggregation problem}  (\mlap) captures this cost-latency trade-off. In \mlap, requests arrive over time at the nodes of a rooted tree representing the hierarchical structure of the system. A service operation transmits a subtree including the root, with cost proportional to the weight of its edges or nodes.  All pending requests within the transmitted subtree are simultaneously and immediately served, modeling the benefit of aggregation. We focus on two primary variants of this online problem, defined by how the latency requirements are modeled: 

\begin{itemize}[nosep]
    \item \textit{MLAP with Deadlines}  (\mlapd): Each request has a hard deadline and the objective is to minimize the total service cost.
    \item \textit{MLAP with Delays} (\mlapdelay): Requests incur delay costs while they remain pending and the objective is to minimize the sum of the total service cost and the total delay cost.
\end{itemize}

Both variants of \mlap~were introduced by \cite{bienkowski2016online}, who formalized the online model and presented the first algorithms with competitive ratios exponential in the tree depth~$D$. This work initiated an active line of research on both \mlapdelay~and \mlapd. For \mlapd, \cite{buchbinder2017depth} gave an $O(D)$-competitive algorithm on 3-decreasing\footnote{A tree is \emph{3-decreasing} if the cost of any node is at most one third that of its parent.} trees, and showed that the result extends to general trees with only a constant-factor loss. This was later strengthened by \cite{mcmahan2021d}, who achieved a $D$-competitive algorithm on arbitrary trees. For the general delay setting, \cite{azar2019general} introduced a unified framework for online problems with delays, yielding an $O(D^2)$-competitive algorithm for \mlapdelay.\footnote{ \cite{talmonthesis} independently obtained an $O(D^3)$-competitive algorithm for \mlap, extending the approach of \cite{buchbinder2017depth}.}

While these developments represent substantial progress, a fundamental gap has remained. Since the $O(D^2)$ bound established by \cite{azar2019general}, the competitive ratio for the general delay setting has seen no asymptotic improvement. In this paper, we close the gap between this quadratic bound and the linear bound of the deadline setting by presenting a $2D$-competitive algorithm for \mlapdelay~with general delay functions.

\subsection{Our Results}

We present new algorithms and analyses with improved competitive ratios for \mlapd and \mlapdelay.  Although the latter problem generalizes the former, we present both results separately for ease of exposition; both results share many key technical insights.

\begin{theorem}
\label{thm:deadlines-main}
    There is a $D$-competitive deterministic online algorithm for \mlapd on arbitrary weighted trees of depth $D$.
\end{theorem}

We note that unlike some previous work on \mlapd~(e.g., \cite{buchbinder2017depth}) our algorithm applies directly to arbitrary trees rather than requiring a hierarchically-decreasing cost structure. Our analysis is tight and there exist instances where the cost incurred by the algorithm is $D$ times the optimal cost. We provide such an instance in Section \ref{sec:tight_example}. 

Since the offline variants of \mlap\ and the \emph{joint replenishment problem} (the special case of \mlap\ with $D=2$) admit primal-dual algorithms~\cite{levi2004primal}, it is natural to hope for an online algorithm based on a primal-dual analysis using the framework of \cite{buchbinder2009design}. However, to quote \cite{bienkowski2020online}, ``\emph{\ldots despite substantial effort of many researchers, the online multi-level setting remains wide open. This is perhaps partly due to impossibility of direct emulation of the cleanup phase in primal-dual offline algorithms in the online setting, as this cleanup is performed in the reverse time order.}''
The primary contribution of our work lies in the analysis via a novel dual fitting framework.

\begin{theorem}
\label{thm:delays-main}
There is a $2D$-competitive deterministic online algorithm for \mlapdelay on arbitrary weighted trees of depth $D$.
\end{theorem}

 Theorem \ref{thm:delays-main} provides the first $O(D)$-competitive algorithm for \mlapdelay, strictly improving upon the previous best bound of $O(D^2)$~\cite{azar2019general}. This shows that the general delay functions do not make the problem substantially harder than the simpler deadlines setting. We emphasize that our algorithm and analysis do not make any assumptions about the tree structure or delay functions other than non-negativity of the instantaneous delays.

\subsection{Overview of Techniques}

\paragraph{Recursive buying algorithm.}
At a high level, the algorithm (for both \mlapd\ and \mlapdelay) initiates a transmission whenever some requests become {\em critical}---i.e., when their deadline arrives or their accumulated delay suffices to pay for the transmission. The main algorithmic question is whether to include additional {\em non-critical} requests in the same service, and if so, which ones. Following a standard approach in the MLAP literature, each vertex $v$ added to the service may spend up to its own cost in ``buying'' further vertices to include in the transmitted tree. Intuitively, we want to avoid transmitting $v$ again soon, so the algorithm attempts to incorporate pending requests within $v$'s subtree that are not too costly to add. This buying procedure is applied recursively, potentially investing up to $D$ times the cost of the critical subtree. The goal of our analysis is then to show that the total cost of all the critical subtrees sent by the algorithm is no more than the hindsight optimum. 

\paragraph{Algorithmic contribution: Consistent priority order over requests.}
While recursive buying has appeared in prior \mlapdelay\ and \mlapd\ algorithms, our approach differs in how it establishes priorities among pending requests. In the deadlines setting, sorting by deadlines yields a globally consistent order. For general delay functions, however, priorities must depend on how quickly a request's delay grows relative to its cost of inclusion: more expensive requests should be allowed to accumulate greater delay before being served. This dependence makes priorities dynamic---changing as the algorithm's state evolves and subtrees are incrementally bought---so that, for example, one request may be prioritized over another in one iteration and vice versa in the next, with neither fully purchased.

Our key innovation is to impose a \emph{consistent} priority order. When a vertex in the current transmission considers which descendants to buy, it assumes none of them has yet been bought, thereby overestimating their inclusion cost (see definitions in Section~\ref{sec:delays}.) 
This simple rule induces a stable ordering---consistent across recursive calls and persistent across successive services---resolving the ambiguity present in previous formulations.

\paragraph{Analytic contribution: Dual fitting.}
Departing from prior analyses, we bound the cost of our algorithm via a primal-dual approach. The dual assigns “charges’’ to requests over time, subject to two constraints: (1) a request’s accumulated charge cannot grow faster than its delay function, and (2) for any subtree, the total future charge of all currently active requests within it cannot exceed the subtree’s cost. The latter ``dual configuration'' constraint is especially challenging, since the algorithm may serve different subsets of these active requests at different times, effectively paying for the same subtree multiple times, whereas the dual can charge it only once. Satisfying this constraint while extracting the cost of all of the critical subtrees from the dual forms the technical core of our paper.

Two novel ideas underlie our analysis. First, we construct a feasible dual solution {\em in hindsight}. Although the recursive buying procedure can be viewed as generating dual variables online, this ``on-the-fly'' dual is not feasible---it cannot anticipate which future dual charges will load the same configuration constraint.
In contrast, the hindsight dual determines how to allocate costs to dual variables based on the full set of configuration constraints each request ultimately participates in. For every vertex served by the algorithm, we first identify the set of requests that caused this vertex to be served, and then carefully select a subset of these requests to act as {\em patrons} to which we can charge the cost of this vertex.

Our second idea establishes dual feasibility by constructing a {\em time-dependent} packing of the dual variables into the tree costs, in order to satisfy the configuration constraints. Specifically, we assign each dual charge on a request to one of its ancestors so that no vertex receives an assignment exceeding its cost. This assignment evolves over time---a crucial feature for maintaining feasibility as new requests appear and are served. Figure~\ref{fig:deadlines_example}  in Section~\ref{sec:deadlines} and the accompanying example illustrate how this time-varying assignment operates in the special case of \mlapd. 

Together, these ideas yield a conceptually simple algorithm accompanied by a clean, tight analysis.

\paragraph{Organization.} Section~\ref{sec:prelim} sets up the formal definitions for the problems we study, as well as their primal and dual formulations. We describe our algorithm and analysis for \mlapd\ in Section~\ref{sec:deadlines}, illustrating the key ideas and intuition for constructing a hindsight dual and a time-dependent dual packing. We then describe its extension to general delays in Section~\ref{sec:delays} illustrating the priority ordering over requests employed by the algorithm, and the significantly more complicated dual packing. Full proofs are provided in Sections~\ref{sec:deferred-proofs-delays} and \ref{sec:extra_deadlines} respectively. Section~\ref{sec:tight_example} provides a tight example for our algorithm and analysis. Further discussion of related work is deferred to Section~\ref{sec:related}.

%% file: prelim.tex
\section{Preliminaries}
\label{sec:prelim}

\subsection{Multi-Level Aggregation Problem (MLAP)}

The \emph{Multi-Level Aggregation Problem} (MLAP) is defined by a tree $\mc{T}$ rooted at node $r$ with weights $w(v)$ on vertices $v\in \tree$, and a set  $\mc{J}$ of job requests. Each request $j \in \mc{J}$ arrives at a node $v_j$ at time $a_j$. The goal of the algorithm is to service all of the requests. A {\em service} consists of ``transmitting'' a subtree $T$ of $\tree$ rooted at $r$ at some time $t$. We say that a request $j$ is \emph{served at time $t$} if $t$ is the first time after $a_j$ such that the subtree transmitted by the algorithm at time $t$ contains $v_j$; let $\served(j)$ denote the time $t$.  A solution to the \mlap problem is a sequence $(T_1, t_1), (T_2, t_2), \dots$ of services, with subtree $T_i$ transmitted at time $t_i$. 

We use $w(T)=\sum_{v\in T} w(v)$ to denote the total weight of a (sub)tree $T$. For vertices $u, v\in \tree$, $P(u,v)$ denotes the unique path between $u$ and $v$; we use $P(v)$ as shorthand for $P(v,r)$, and $P(j)$ as shorthand for $P(v_j, r)$.

\paragraph{MLAP with deadlines.} In the {\em deadline} version of the problem (\mlapd), each request $j \in\mc{J}$ arriving at time $a_j$ is endowed with a deadline $d_j \geq a_j$ and must be served before its deadline. A solution is feasible if for every request $j$, $\served(j) \in [a_j, d_j]$. The cost of a feasible solution $(T_1, t_1), (T_2, t_2), \dots $ is given by
\[ \costalg \coloneq \treecostalg = \sum_{i} w(T_i).\]
Observe that an algorithm can save cost by serving multiple requests together.

\paragraph{MLAP with delays.} In the more general {\em delay} version (\mlapdelay), each request $j \in\mc{J}$ is endowed with a delay function $\delay(j,t)$, which returns the amount of instantaneous delay incurred by the request $j$ at time $t$.  For any request $j$, we assume that $\delay(j, t) \ge 0$ at all times $t$ and $\delay(j, t) = 0$ for $t\le a_j$. 
A request $j$ can be served at any time in $[a_j, \infty)$ but incurs delay until it is served.  Our objective is to minimize the total tree cost plus the total delay incurred:
\[
   \costalg \coloneq \treecostalg + \delayalg = \sum_{i} w(T_i) + \sum_{j \in \mc{J}} \int_{a_j}^{\served(j)} \delay(j, t) dt.
\]

\paragraph{Discrete vs continuous time.} In the deadline setting, we assume without loss of generality that all events happen at integral times, and the algorithm only needs to consider times in the set $\cup_{j\in\mc{J}}\{a_j, d_j\}$. In the delay setting, we assume that time is continuous. It is easy to reduce an instance with discrete delay functions to the continuous setting without any loss in performance \cite{bienkowski2020online}.

\paragraph{Online arrivals and competitive ratio.}
In online MLAP, the tree $\mc{T}$ is known in advance, but the requests $\mc{J}$ arrive online. When a request $j$ arrives at time $a_j$, all of its attributes (location, deadline, delay function) are revealed to the algorithm. Let $\costopt$ be the cost of the hindsight optimal solution when both $\mc{T}$ and $\mc{J}$ are known upfront. Then our goal is to minimize the competitive ratio $\max_{\mc{T}, \mc{J}} \frac{\costalg}{\costopt}$.

\paragraph{Accessing the delay functions.} For convenience, we overload the delay notation in several ways. For a request $j$ and time interval $I$,  $\del(j,I)$ denotes the delay accumulated by request $j$ in time interval $I$:
\[\del(j,I) \coloneq \int_{I \cap [a_j, \infty)} \delay(j,t) dt.\]
For a set $Q \subseteq \mc{J}$ of requests and interval $I$, $\del(Q,I)$ denotes the total delay accumulated by all $j\in Q$ over $I$: $\del(Q, I) = \sum_{j\in Q} \del(j,I)$. Finally, for any time $t$, we use $[t]$ to indicate the time interval $[0, t]$ and thus $\delay(j, [t])$ denotes the delay accumulated by job $j$ until time $t$.
Our algorithm will assume access to an oracle that can answer two kinds of queries: (1) given a request $j$ and time $t$, it returns $\del(j, [a_j,t])$; and (2), given a request $j$, a start time $t_1$, and a target $\gamma$, it returns an end time $t_2$ such that $\del(j, [t_1, t_2])=\gamma$.

\paragraph{Other notation.} For a vertex $v$, we use $\tree_v$ to denote the entire subtree of $\tree$ rooted at $v$. Likewise, for a $T \subset \tree$ we use $T_v$ to denote the entire subtree of $T$ rooted at $v$. For a set $Q \subseteq \mc{J}$ of requests, we use $Q_u\subseteq Q$ to denote the subset of requests located at vertices in $\tree_u$, and $\tree_u^Q$ to denote the subtree spanned by $u$ and the locations of requests in $Q_u$.

\subsection{Primal and Dual LPs}

We analyze our algorithm through the configuration linear program (\lp) relaxation and its dual, using the framework of dual fitting. We first present the \mlapd\  formulation, which serves as the structural foundation. We then extend it to the general delay function setting by introducing variables that account for accumulated delay costs, reusing the notation and variables from the deadline case for consistency. 

Let $\mathcal{S}$ be the set of all subtrees containing the root, and let $x(S, t)$ be the fraction of tree $S$ transmitted at time $t$. The primal objective minimizes the total cost of transmitted subtrees while ensuring every request is fully served.

\begin{figure}[h]
\centering
\setlength{\tabcolsep}{6pt}
\renewcommand{\arraystretch}{1.2}

\begin{tabular}{@{}l l @{\qquad} | l l@{}}
\multicolumn{2}{c}{\textsc{Primal LP for \mlapd}} & \multicolumn{2}{c}{\textsc{Dual LP for \mlapd}}\\[12pt]

$\min$\ &
$\displaystyle \sum_{S \in \mc{S}} \sum_{t} \cost(S)\,x(S,t)$
&
$\max$\ &
$\displaystyle \sum_{j \in \mc{J}} \alpha_j$
\\[5pt]

\text{s.t.}\ &
{\small
$\displaystyle
\sum_{t \in [a_j, d_j]} \sum_{S:\, v_j\in S} x(S,t) \ge 1,\ \forall j \in \mc{J}
$}
&
\text{s.t.}\ &
{\small
$\displaystyle
\sum_{\substack{j:\, t \in [a_j, d_j]\\ v_j\in S}} \alpha_j \le \cost(S),\ \forall S \in \mc{S},\ \forall t
$}
\\[3pt]

& $\displaystyle x(S,t) \ge 0$
&
& $\displaystyle \alpha_j \ge 0$
\end{tabular}
\end{figure}

We say that a request $j$ is \emph{active} during the interval $[a_j, d_j]$. The dual objective maximizes the total value of variables~$\alpha_j$, subject to the condition that at any time $t$ and for any rooted subtree $S$, the sum of dual variables for active requests in $S$ at time $t$ does not exceed the cost of $S$.

The delay setting extends this formulation by incorporating delay costs into the objective; we assume that the time is continuous. We introduce a new variable $y(j, t)$ to denote the fraction of request $j$ that is left ``unsent'' at time $t$. The primal objective now minimizes the sum of service costs and delay, and the primal constraints ensure consistency between the $x$ and $y$ variables.

\begin{figure}[h]
\centering
\setlength{\tabcolsep}{6pt}
\renewcommand{\arraystretch}{1.2}
\begin{small}
\begin{tabular}{@{}l l @{\qquad} | l l@{}}
\multicolumn{2}{c}{\textsc{Primal LP for \mlapdelay}} & \multicolumn{2}{c}{\textsc{Dual LP for \mlapdelay}}\\[12pt]
$\min$\ & $\displaystyle \sum_{S \in \mc{S}} \int \cost(S)\,x(S,t)\,dt$ $\displaystyle +\; \sum_{j \in \mc{J}} \int \delay(j,t)\,y(j,t)\,dt$
& $\max$\ & $\displaystyle \sum_{j \in \mc{J}} \int \alpha_{j t}\,dt$ \\[3pt]

\text{s.t.}\ &
{\small
$\displaystyle
y(j,t) + \sum_{S:\, v_j\in S} \int_{a_j}^{t} x(S,\tau)\,d\tau \ge 1,\ \forall j \in \mc{J},\ \forall t \ge a_j
$}
&
\text{s.t.}\ &
{\small
$\displaystyle
\sum_{\substack{j:\, a_j \le t,\\ v_j\in S}}
\int_{t}^{\infty} \alpha_{j \tau}\,d\tau \le \cost(S),\ \forall S \in \mc{S},\ \forall t
$}
\\[3pt]

& $\displaystyle x(S,t) \ge 0$ 
& &
$\displaystyle \alpha_{j t} \le \delay(j,t),\ \forall j \in \mc{J},\ \forall t \ge a_j$ \\[3pt]

& $\displaystyle y(j,t) \ge 0$
& &
$\displaystyle \alpha_{j t} \ge 0$
\end{tabular}
\end{small}
\end{figure}

The dual objective maximizes the total value of variables $\alpha_{jt}$ as before. The first constraint for any time $t$ and any rooted subtree $S$ is similar to the constraint in the dual of \mlapd, but now considers the {\em future} contributions (over times $\tau\ge t$) of the requests in $S$ that are active at time $t$.
The second constraint limits each dual variable~$\alpha_{jt}$ by the instantaneous delay rate of request $j$ at time $t$.

%% file: deadlines_v2.tex
\section{MLAP with Deadlines}
\label{sec:deadlines}

In this section, we present our Algorithm \ref{alg:multilevel_agg} for MLAP with deadlines.  This algorithm is similar in spirit to previous ones \cite{azar2019general, buchbinder2017depth, mcmahan2021d} that are known to obtain $O(D)$-competitive ratios but  unifies and simplifies their ideas.  Moreover, it does not require the tree $\tree$ to satisfy any special structure. We describe the algorithm in Section~\ref{sec:deadline-alg} and then introduce a novel dual fitting based analysis for it in Section~\ref{sec:deadline-duals}. All proofs are deferred to Section~\ref{sec:extra_deadlines}. 

\subsection{Algorithm}
\label{sec:deadline-alg}

We call a request $j$ {\em pending} if it has arrived but is yet to be served. We say that a request is {\em critical} (resp., critical in $\tree_v$) if it has the earliest deadline among all pending requests in $\tree$ (resp., $\tree_v$).
 
At a high level, our algorithm sends a tree whenever the deadline of a pending request is reached. Besides sending the critical request in this tree, the algorithm faces a choice in whether or not to send additional requests early (i.e., before their deadlines). Sending a request $j$ early  saves the algorithm the cost of any ancestors of $j$ that are already included in the serviced tree; but if a new request arrives again at the same node $v_j$, the added expense of including $j$ in the service goes to waste. The key idea towards striking a balance between these two possibilities is to follow a ski-rental type approach: early requests are included if their cost can be ``charged'' to vertices already included in the tree.

Specifically, whenever the algorithm includes a vertex $v$ in the tree $T$ to be served, it assigns to $v$ a budget $b_v$ equal to its cost $w(v)$. It spends this budget through an ``explore'' procedure run from every vertex $v\in T$. During this procedure, $v$ considers pending requests in its subtree in increasing order of deadline. For each such request, $v$ identifies the \emph{frontier node} on the path from $v$ to that request, i.e. the closest descendant of $v$ not yet in $T$ whose parent is already in $T$, and invests its budget toward paying off that frontier node's cost. Any frontier node that is fully paid off is added to the tree and recursively explored with its own fresh budget, and any partial payments are saved for the next iteration (via the counters $\{c_u\}$).
The process ends when there are no pending requests remaining in the subtrees of any nodes with leftover budget, at which point the tree is sent and any leftover budget is discarded. We give a detailed pseudocode in Algorithm \ref{alg:multilevel_agg}. Our main goal in this section is to prove Theorem \ref{thm:deadlines-main}. This result follows directly from Lemmas \ref{lem:alg-bound-deadlines} and \ref{lem:deadlines-opt-lb}, stated below, which provide upper and lower bounds on the cost of Algorithm \ref{alg:multilevel_agg} and the optimal solution, respectively.

\begin{algorithm}
\caption{Online Multilevel Aggregation with Deadlines}
\begin{minipage}{\linewidth}
\label{alg:multilevel_agg}
\begin{algorithmic}[1]
\State \textbf{Initialization:}
\State Initialize $b_v, c_v \leftarrow 0$ for every node $v \in \tree$ \Comment{{\footnotesize $c_v$ denotes the investment received by $v$ and $b_v$ its budget.}}
\Statex
\Function{\UponCritical()}{}\footnote{Note that multiple requests could have the same deadline and become critical at the same time $t$. We resolve ties following a fixed tie-breaking rule. The function \UponCritical() could thus potentially be triggered multiple times at $t$ until no critical requests remain.} \Comment{{\footnotesize The deadline of a critical request $j$ is reached. $P(v_j)$ is the critical path.}}
    \State $j \gets $ critical request
    \State $T \gets \emptyset$
    \For{$u \in P(v_j)$}
        \State $c_u \gets 0$
        \State \Call{\Add}{$u, T$}
    \EndFor
    \For{$u \in P(v_j)$}
        \State \Call{\Explore}{$u, T$}
    \EndFor
    \State transmit $T$
\EndFunction
\Statex
\Function{\Explore}{$v, T$}
\Comment{{\footnotesize Add $v$ and some of its descendants to the serviced tree $T$.}}
    \State $b_v \gets w(v)$
    \While{$b_v \neq 0$ \textbf{and} there remain pending requests in $\tree_v$}
        \LetState{$H$}{$\{ u \in \tree_v \ | \  u \notin T\ \&\  \mathrm{parent}(u) \in T\}$} 
        \LetState{$j'$}{pending request in $\tree_v$ with the earliest deadline}
        \LetState{$v'$}{node in $H$ on path between $r$ and $j'$} 
        \State \Call{\Invest}{$v, v'$}
        \If{$c_{v'} = w(v')$}
            \State $c_{v'} \leftarrow 0$
            \State \Call{\Add}{$v', T$}
            \State \Call{\Explore}{$v',T$}
        \EndIf
    \EndWhile 
\EndFunction
\Statex
\Function{\Add}{$v, T$}
    \State $T \gets T \cup \{v\}$
        \State Mark all pending requests on $v$ as served
\EndFunction

\Statex
\Function{\Invest}{$v, v'$}
    \LetState{$y$}{$\min(b_v, w(v') - c_{v'})$}
    \State increase $c_{v'}$ by $y$
    \State decrease $b_v$ by $y$
\EndFunction

\end{algorithmic}
\end{minipage}
\end{algorithm}

\subsection{Competitive Ratio Analysis}
\label{sec:deadline-duals}
Recall that any tree sent by the algorithm is initialized by a request becoming critical.  We call the path connecting this request to the root the \emph{critical path} of the tree. We break up our analysis into two parts. First, we relate the cost of the algorithm to the cost of the critical path in every tree scheduled for service by the algorithm. Then we charge this critical path cost to a dual constructed in hindsight.

\paragraph{Primal analysis.}
Intuitively, whenever a vertex $v$ in the critical path spends its budget towards the cost of its descendants, any descendants added to the tree further pass the same budget down to their own descendants, and so on until the budget reaches the leaves of the tree. In essence, any budget created by the critical path gets spent up to $D$ times, where $D$ is the depth of the tree. 

Formally, at any time $t$ during the execution of the algorithm, let $c_v(t)$ denote the value of the counter $c_v$ at time $t$, i.e., the total investment received by $v$ from its ancestors since the last time $v$ was sent. Let $\hat{w}(v, t) = w(v) - c_v(t) \geq 0$; we will refer to $\hat{w}(v, t)$ as the ``unpaid'' cost of $v$ at time $t$. Then the following lemma relates the algorithm's cost to the total unpaid cost of the critical paths.  We skip the proof, but instead prove a generalization (Lemma~\ref{lem:alg-bound-delays}) to MLAP with delays in Section~\ref{sec:primal-deferred}.

\begin{lemma}
\label{lem:alg-bound-deadlines}
Let $(T_1, t_1), \dots, (T_k, t_k)$ be the solution generated by Algorithm~\ref{alg:multilevel_agg}, where $t_i$ denotes the time at which tree $T_i$ is sent. 
Let $S_i \subset T_i$ be the critical path of tree $T_i$. Then, 
\[
\costalg \coloneq \treecostalg \leq D \cdot \sum_{i=1}^k \sum_{v \in S_i} \hat{w}(v, t_i).
\]
\end{lemma}

\paragraph{Dual analysis.}
We now construct a feasible solution to the \mlapd\ dual LP, which recovers the total unpaid cost of the critical paths. By weak LP duality, the cost of an optimal solution is at least the objective of any feasible dual solution.

\begin{lemma}
\label{lem:deadlines-opt-lb}
Let $(T_1, t_1), \dots, (T_k, t_k)$ be the solution generated by Algorithm~\ref{alg:multilevel_agg}, and $S_i \subset T_i$ be the critical path of tree $T_i$. Then, 
$$\costopt \geq \sum_{i=1}^k \sum_{v \in S_i} \hat{w}(v, t_i).$$
\end{lemma}
Before describing the dual construction, we first give some high-level intuition. Consider a tree $T_i$ sent by the algorithm with critical request $j$. It seems reasonable to charge the entire cost of the critical component $S_i\subset T_i$ to the dual variable $\alpha_j$. However, suppose that $j$ was pending at the previous time $t_{i-1}$ when a tree $T_{i-1}$ was sent, and let $j'$ denote the critical request in $T_{i-1}$. Then, the two dual variables $\alpha_j=w(S_i)$ and $\alpha_{j'}=w(S_{i-1})$ collectively violate the constraint corresponding to $S_i\cup S_{i-1}$ at time $t_{i-1}$. In particular, $w(S_i\cup S_{i-1})$ is strictly smaller than $w(S_i)+w(S_{i-1})$ because the two critical components share the root $r$ (and possibly other nodes). Therefore, in order to maintain feasibility, request $j$'s dual cannot pay for the root's cost, $w(r)$, in $T_i$. 

To remedy this situation, we recall that the root $r$ invested its budget in some of its descendants during $\Explore(r,T_{i-1})$. Since the request $j$ was pending at time $t_{i-1}$, but did not get included in $T_{i-1}$'s service, this budget was spent entirely towards other requests. We will accordingly try to recover $r$'s cost through duals for these requests. 

We formalize this idea through a recursive procedure called $\Charge$ (see Algorithm~\ref{alg:deadlines-analysis}). For a vertex $v$ contained in tree $T$, $\Charge$ recovers (a portion of) the cost of $v$ by incrementing the dual solution. The procedure either directly charges the critical request in $\tree_v$ that caused $v$ to be included in the service, or it considers the previous time $v$ executed the $\Explore$ procedure and recursively calls $\Charge$ on all of its descendants that received investments from $v$. In either case, we show that this process increases the total objective value of the dual solution by an adequate amount. 

\paragraph{Assignment of duals towards node costs.} Recall that dual feasibility requires that for every rooted subtree $S$ and time $t$, the sum over the duals $\alpha_j$ for requests $j$ in $S$ that are active at $t$ is at most the cost of $S$. In order to argue feasibility it is convenient to assign each increment in $\alpha_j$ towards the cost of some ancestor $u$ of $v_j$. Ideally, we would like to ensure that no vertex $u$ receives more than its total weight from active requests in its subtree $\tree_u$. Since the set of active requests changes over time, however, this assignment must also evolve over time. We use $\delta_{utj}$ to denote the amount assigned by request $j$ at time $t$ to an ancestor $u$ of $v_j$. Our goal is to choose these variables so that, at every time $t \in [a_j,d_j]$, each request $j$ assigns total mass $\alpha_j$ to its ancestors, and the quantities $\delta_{utj}$ at time $t$ pack into the tree costs. The example below illustrates how this assignment is updated over time.

\paragraph{Example of the dual construction.}

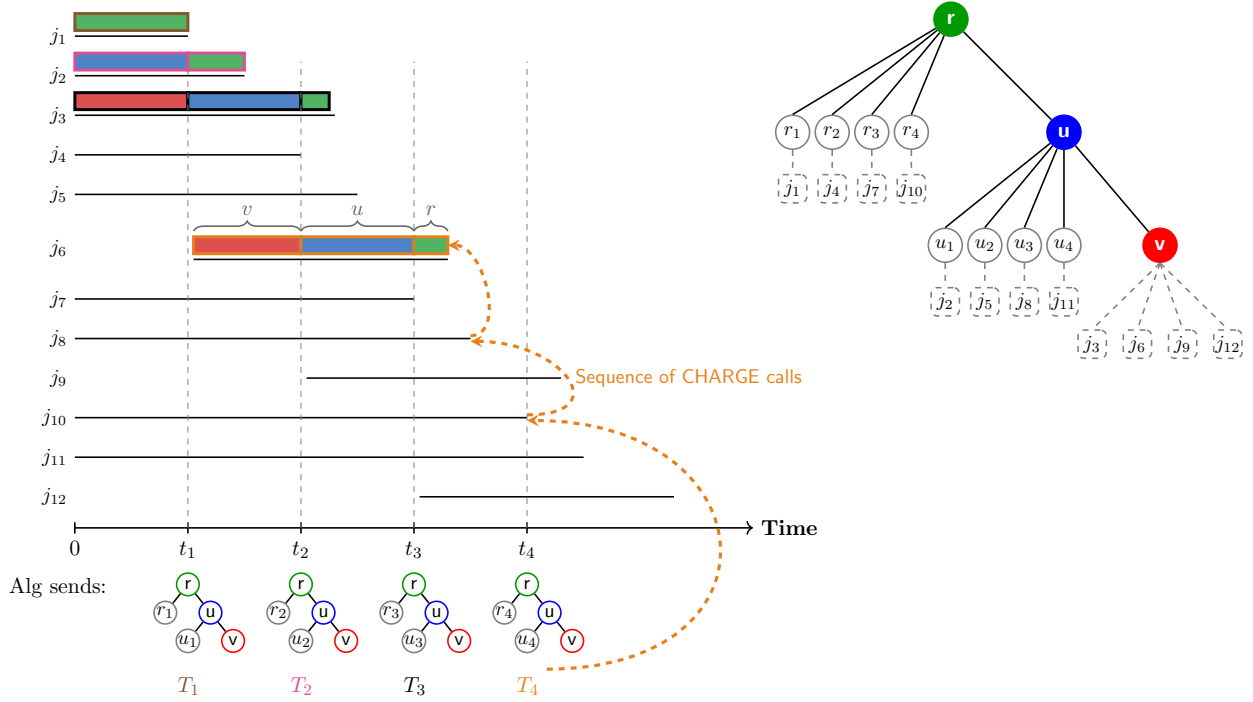
\begin{figure}[h]
    \centering
    \begin{adjustbox}{width=1\textwidth}
    \input{mlap_deadlines_tikz_horizontal}
    \end{adjustbox}
    \caption{An example of the dual construction for the MLAP deadline setting. On the right, tree nodes are depicted by circles, and requests are depicted by dashed squares. All tree nodes have a cost of $0$, except for nodes $r, u, \text{ and } v$, which have a cost of $1$. On the left, each horizontal line segment represents a request $j_k$ arriving at its corresponding node, spanning from its arrival time (left endpoint) to its deadline (right endpoint). The algorithm sends four subtrees ($T_1, T_2, T_3, T_4$) at times indicated by the vertical dashed lines. Read the accompanying text for an explanation of the highlighted colored segments.} 
    \label{fig:deadlines_example}
\end{figure}

We illustrate the dual construction using the example shown in Figure~\ref{fig:deadlines_example}. In this example, the input tree $\mc{T}$ has three primary nodes $r,u,v$, each of cost $1$. In addition, $r$ has zero-cost children $r_i$, and $u$ has zero-cost children $u_i$. Requests $j_1,\ldots,j_{12}$, shown as dashed squares, arrive at their corresponding nodes in this order. On the left side of the figure, the solid black horizontal segment associated with each request runs from its arrival time to its deadline. Requests $j_6,j_9$, and $j_{12}$ arrive shortly after the algorithm sends a tree, while all other requests arrive at time $0$. The deadlines are chosen so that, whenever the algorithm sends a tree, the requests located at the children $r_i$ of $r$ have the earliest deadlines, followed by the requests located at the children $u_i$ of $u$. The request at $v$ has the next earliest deadline within the subtree rooted at $u$. Algorithm~\ref{alg:multilevel_agg} sends four trees, one at each dashed vertical line, with each tree serving three requests.

We now describe the primal recursive buying procedure and the corresponding dual variables ${\alpha_j}$. In each tree sent by Algorithm~\ref{alg:multilevel_agg}, the request located at some child $r_i$ is the critical request. In tree $T_i$, the root first explores and adds the request at $u_i$ to the tree. The vertex $u$ then explores and adds the request at $v$, after which the tree $T_i={r,r_i,u,u_i,v}$ is scheduled. The only node whose cost must be paid for by the dual in each tree is the root, whose unpaid cost is always equal to its full cost.

For the first tree $T_1$, we increase the dual of request $j_1$ to equal the root's cost. For the second tree, it is natural to increase the dual of the critical request $j_4$. However, note that request $j_4$  overlaps with $j_1$ and incrementing $\alpha_{j_4}$ would violate the dual constraint for $S = \{r, r_1, r_2\}$. Thus, rather than charging the cost of $T_2$ to $j_4$, we charge it to a node already bought by $T_1$, namely node $u$. The critical request under $u$ in $T_1$ is $j_2$, so we increase the dual of $j_2$ to pay for the root in $T_2$. We apply the same recursive charging process for $T_3$ and $T_4$, except that these trees require two and three levels of recursion, respectively. In particular, $T_3$ charges request $j_3$, while $T_4$ charges  $j_6$.

The highlighted segments in the figure indicate requests with non-zero duals, as well as the time intervals over which these dual variables participate in their respective dual constraints. The boundary color of each segment identifies the tree that initiates the corresponding charge. For example, the pink boundary on the box for $j_2$ indicates that $j_2$ is charged on behalf of tree $T_2$.

We now argue that the dual solution is feasible. Each request with a nonzero dual assigns its dual value to one of its ancestors at every time at which it is active. Request $j_1$ has only one costly ancestor, namely $r$, and assigns its dual to $r$ throughout the interval $[0,d_1]$. Request $j_2$ cannot assign its dual to $r$ during $[0,d_1]$, and therefore assigns it instead to its ancestor $u$ until time $t_1$, and to $r$ thereafter. This assignment is feasible because $d_1\le t_1$. The assignment is illustrated in Figure~\ref{fig:deadlines_example} by the fill colors of the corresponding dual boxes. Similarly, $j_3$ cannot assign its dual to either $r$ or $u$ during the interval $[0,t_1]$, and so assigns it to $v$; it then assigns it to $u$ during $[t_1,t_2]$, and to $r$ thereafter.

As Figure~\ref{fig:deadlines_example} shows, highlighted segments of the same color never overlap. This non-overlap property is the key to feasibility: at any time and for any node, the total assignment received from active requests in its subtree is never larger than the cost of that node.

We formally construct our dual in Algorithm \ref{alg:deadlines-analysis}. Lemma \ref{claim:dual_obj_inc} establishes that our dual construction achieves the desired objective value. See Section \ref{sec:extra_deadlines} for a proof. Lemma \ref{lem:deadlines-opt-lb} then follows from dual feasibility, which we argue next.

\begin{restatable}{lemma}{dualobjinc}
\label{claim:dual_obj_inc}
Each call to $\Charge(u, T, \gamma)$ increases the dual objective, $\sum_{j\in\mc{J}} \alpha_j$, by $\gamma$.
\end{restatable}

\paragraph{Dual feasibility.}
To argue dual feasibility, we introduce auxiliary variables
that track how each increment to $\alpha_j$ is distributed
among the ancestors of~$v_j$. Recall that each call to
$\Charge(u, T, \gamma)$ returns a set~$B$ of
(request, amount) pairs: each pair $(j, \gamma') \in B$
records that $\alpha_j$ was increased by $\gamma'$, either
directly in this call or through a recursive sub-call.

\begin{definition}[Assignment variables]
\label{def:delta}
All variables $\delta_{u\tau j}$ are initialized to $0$.
For each call $\Charge(u, T, \gamma)$ executed during
Algorithm~\ref{alg:deadlines-analysis}, let $t$ be the
transmission time of~$T$, let $t_{\prev}$ be the last
time~$u$ was transmitted before~$t$, and let $B$ denote
the set of pairs returned by this call. For each
$(j, \gamma') \in B$ and each $\tau \in (t_{\prev}, t]$,
we increment
\[
  \delta_{u\tau j} \;\leftarrow\; \delta_{u\tau j} + \gamma'.
\]
We define
$\delta_{u\tau} \coloneq \sum_{j} \delta_{u\tau j}$
as the total assignment to node~$u$ at time~$\tau$.
\end{definition}

Intuitively, each call $\Charge(u, T, \gamma)$ routes
dual value through node~$u$ during the interval
$(t_{\prev}, t]$. The variable $\delta_{u\tau j}$ records
how much of request~$j$'s dual is assigned to node~$u$
at time~$\tau$, and $\delta_{u\tau}$ tracks the total
capacity of~$u$ consumed at time~$\tau$.

The following lemma formalizes the relationships between these variables. 
\begin{restatable}{lemma}{deadlinesdeltaalpha}
\label{lem:deadlines-delta-alpha-3}
Algorithm~\ref{alg:deadlines-analysis} satisfies the following properties:
    \begin{enumerate}[label=(\alph*)]
        \item For every $j$ and $t\in [a_j, d_j]$, $\sum_{u\in P(v_j)}\delta_{utj} = \alpha_j$.
        \item For all $v$ and $t$, $\delta_{vt}\le w(v)$.
    \end{enumerate}
\end{restatable}

We prove the lemma in Section~\ref{sec:extra_deadlines}. The key technical component in our proof (formalized as Claim~\ref{claim:ddl_v_less_than_ddl_u} below and proved formally in Section~\ref{sec:extra_deadlines}) is to argue that if a sequence of nested calls $\Charge(u_{(1)}, T_{(1)}, \cdot) \rightarrow \Charge(u_{(2)}, T_{(2)}, \cdot) \rightarrow \cdots$ culminates in the increase of a dual variable $\alpha_j$, then it holds that the deadline $d_j$ of $j$ is no more than the time at which the tree $T_{(1)}$ corresponding to the first call is scheduled. The reader can verify this property visually in Figure~\ref{fig:deadlines_example} -- for example, we charge the cost of tree $T_4$ to request $j_6$; the deadline of $j_6$ is before the time $t_4$ at which $T_4$ is scheduled. The property is a consequence of prioritizing investment in requests with earlier deadlines over those with later deadlines. 

\begin{restatable}{claim}{deadlineorderclaim}\label{claim:ddl_v_less_than_ddl_u}
Suppose $\Charge(u, T, \gamma)$ calls $\Charge(v, T', \gamma')$ in Algorithm \ref{alg:deadlines-analysis}. Then the deadline of the critical request $j'$ of $T'_{v}$ is at most the deadline of the critical request $j$ of $T_u$.
\end{restatable}

\begin{algorithm}[!htbp]
\caption{Dual Construction for \mlapd}
\label{alg:deadlines-analysis}
\begin{algorithmic}[1]
\Statex
\For{$i = 1, 2, \ldots$}
\LetState{$T_i$}{tree transmitted at time $t_i$}
\LetState{$S_i$}{critical path in $T_i$}
\For{$u \in S_i$}
\State \Call{\Charge}{$u$, $T_i$, $\hat{w}(u, t_i)$} \Comment{{\footnotesize $u$ charges its unpaid cost in tree $T_i$ to a subset of requests}}
\label{line:top_level_recursive}
\EndFor
\EndFor
\Statex
\Function{\Charge}{$u$, $T$, $\gamma$} \Comment{{\footnotesize $u$ is transmitted in tree $T$, and charges $\gamma$ to a subset of requests}}
    \LetState{$t$}{time when $T$ was transmitted}
    \LetState{$T_{\prev}, t_{\prev}$}{last tree where $u$ was transmitted (strictly before $t$) and its time}
    \Statex
    \LetState{$j$}{critical request in $T_u$}  \label{line:critical_j}
    \If{$a_j > t_{\prev}$}
        \State set $\alpha_{j} \leftarrow \alpha_j + \gamma$ \Comment{{\footnotesize $u$ charges the entire $\gamma$ amount to $j$}}
        \State \Return{$\{(j, \gamma)\}$}
    \Else
        \LetState{$A_u$}{Set of nodes that $u$ invests in during $T_{\prev}$} 
        \Statex \Comment{{\footnotesize $u$ will recursively charge requests through proxy nodes in $A_u$}}
        \For{$v \in A_u$}
            \LetState{$y(u, v, T_{\prev})$}{Amount invested by $u$ in $v$ in tree $T_{\prev}$}
            \LetState{$T_{\nxt}^v$}{next tree where $v$ is transmitted (inclusive of $T_{\prev})$}
            \State $B_v \gets $ \Call{\Charge}{$v$, $T_{\nxt}^v$, $y(u, v, T_{\prev}) \cdot \frac{\gamma}{w(u)}$} 
            \Statex \Comment{{\footnotesize The amount charged by $v$ is proportional to the investment $u$ makes in $v$}}
        \EndFor
        \State \Return{$\bigcup_{v \in A_u} B_v$}
    \EndIf
\EndFunction
\end{algorithmic}
\end{algorithm}

To wrap up this section, we show how Lemma~\ref{lem:deadlines-delta-alpha-3} implies dual feasibility.

\begin{lemma}\label{lemma:config_feasibility}
The dual constructed by Algorithm \ref{alg:deadlines-analysis} forms a feasible dual solution.
\end{lemma}
\begin{proof}
Consider any subtree $S$ and time $t$. Let $J(S, t) = \{j : t \in [a_j, d_j] \text{ and } v_j \in S\}$ be the requests that contribute to the LHS of the dual constraint. We now have the following:
\begin{align*}
    \sum_{j \in J(S, t)} \alpha_j = \sum_{j \in J(S, t)} \sum_{u \in P(v_j)} \delta_{utj} \leq \sum_{u \in S} \delta_{ut} \leq \sum_{u \in S} w(u).
\end{align*}
Note that in the first inequality, we use that $\delta_{vtj} = 0$ for all $v \notin P(v_j, r)$ and that $J(S,t) \subseteq J$. The three (in)equalities follow from Lemma~\ref{lem:deadlines-delta-alpha-3}.
\end{proof}

%% file: mlap_deadlines_tikz_horizontal.tex
\begin{tikzpicture}[
    node distance=1.2cm and 1.2cm,
    base_circle/.style={
        circle, 
        draw=black, 
        thick, 
        minimum size=6mm, 
        inner sep=0pt,
        font=\sffamily\bfseries
    },
    base_square/.style={
        rectangle, 
        draw=gray,        
        fill=none,        
        dashed,           
        rounded corners=3pt, 
        thick, 
        minimum width=5mm, 
        minimum height=5mm, 
        inner sep=0pt,
        font=\small\sffamily
    },
    root/.style={base_circle, draw=green!60!black, fill=green!60!black}, 
    mid/.style={base_circle, draw=blue, fill=blue}, 
    leaf/.style={base_circle, draw=red, fill=red},
    root_small/.style={base_circle, draw=green!60!black, font=\small\sffamily, minimum size=4mm},  
    mid_small/.style={base_circle, draw=blue, font=\small\sffamily,         minimum size=4mm}, 
    leaf_small/.style={base_circle, draw=red, font=\small\sffamily,         minimum size=4mm}, 
    small_circle/.style={base_circle, draw=gray, font=\small\sffamily}, 
    edge/.style={-, thick},
    tiny_circle/.style={base_circle, draw=gray, font=\small\sffamily,  minimum size=4mm}, 
]

    \def\xr{13}
    \def\yr{2}

    \node[root] (r) at (\xr+2.5, \yr) {\color{white}{r}};
    
    \node[small_circle] (r1) at (\xr-.3,\yr-2) {$r_1$};
    \node[small_circle] (r2) at (\xr+.4,\yr-2) {$r_2$};
    \node[small_circle] (r3) at (\xr+1.1,\yr-2) {$r_3$};
    \node[small_circle] (r4) at (\xr+1.8,\yr-2) {$r_4$};
    
    \node[base_square] (j1)  at (\xr-.3,\yr-3)  {$j_1$};
    \node[base_square] (j4)  at (\xr+.4,\yr-3)  {$j_4$};
    \node[base_square] (j7)  at (\xr+1.1,\yr-3)  {$j_7$};
    \node[base_square] (j10) at (\xr+1.8,\yr-3){$j_{10}$};

    \draw[edge] (r) -- (r1.north); 
    \draw[edge] (r) -- (r2.north);
    \draw[edge] (r) -- (r3.north); 
    \draw[edge] (r) -- (r4.north);

    \draw[draw=gray, dashed, thick] (r1.south) -- (j1.north);
    \draw[draw=gray, dashed, thick] (r2.south) -- (j4.north);
    \draw[draw=gray, dashed, thick] (r3.south) -- (j7.north);
    \draw[draw=gray, dashed, thick] (r4.south) -- (j10.north);

    \node[mid] (u) at (\xr+4.5,\yr-2) {\color{white}{u}}; 
    
    \node[small_circle] (u1) at (\xr+2.4,\yr-4) {$u_1$};
    \node[small_circle] (u2) at (\xr+3.1,\yr-4) {$u_2$};
    \node[small_circle] (u3) at (\xr+3.8,\yr-4) {$u_3$};
    \node[small_circle] (u4) at (\xr+4.5,\yr-4) {$u_4$};
    
    \node[base_square] (j2) at (\xr+2.4,\yr-5) {$j_2$};
    \node[base_square] (j5) at (\xr+3.1,\yr-5) {$j_5$};
    \node[base_square] (j8) at (\xr+3.8,\yr-5) {$j_8$};
    \node[base_square] (j11) at (\xr+4.5,\yr-5) {$j_{11}$};

    \draw[edge] (u) -- (u1.north); 
    \draw[edge] (u) -- (u2.north);
    \draw[edge] (u) -- (u3.north); 
    \draw[edge] (u) -- (u4.north);
    
    \draw[edge] (r) -- (u); 

    \draw[draw=gray, dashed, thick] (u1) -- (j2.north);
    \draw[draw=gray, dashed, thick] (u2) -- (j5.north);
    \draw[draw=gray, dashed, thick] (u3) -- (j8.north);
    \draw[draw=gray, dashed, thick] (u4) -- (j11.north);

    \node[leaf] (v) at (\xr+6.2,\yr-4) {\color{white}{v}}; 
    
    \node[base_square] (j3) at (\xr+5,\yr-5.75) {$j_3$};
    \node[base_square] (j6) at (\xr+5.8,\yr-5.75) {$j_6$};
    \node[base_square] (j9) at (\xr+6.6,\yr-5.75)  {$j_9$};
    \node[base_square] (j12) at (\xr+7.4,\yr-5.75) {$j_{12}$};

    \draw[edge] (u) -- (v); 

    \draw[draw=gray, dashed, thick] (v.south) -- (j3.north);
    \draw[draw=gray, dashed, thick] (v.south) -- (j6.north);
    \draw[draw=gray, dashed, thick] (v.south) -- (j9.north);
    \draw[draw=gray, dashed, thick] (v.south) -- (j12.north);

        \definecolor{barred}{RGB}{220,80,80}
        \definecolor{barblue}{RGB}{80,130,200}
        \definecolor{bargreen}{RGB}{80,180,100}
        \definecolor{barpurple}{RGB}{160,80,200}
        \definecolor{baryellow}{RGB}{220,200,80}
        \definecolor{barorange}{RGB}{230,140,50}
        
        \def\off{4.3}
        
        \node[left] at (0, 6.0-\off)  {\small $j_1$};
        \node[left] at (0, 5.3-\off)  {\small $j_2$};
        \node[left] at (0, 4.6-\off)  {\small $j_3$};
        
        \draw[thick] (0.0, 6-\off) -- (2, 6-\off); 
        \draw[thick] (0.0, 5.3-\off) -- (3, 5.3-\off); 
        \draw[thick] (0.0, 4.6-\off) -- (4.6, 4.6-\off); 

        \fill[bargreen, draw=brown!70!black, line width=1.5pt]     (0.0, 6.4-\off) rectangle (2.0, 6.1-\off);
        \fill[barblue, draw=magenta!80!gray, line width=1.5pt]   (0, 5.7-\off) rectangle (2, 5.4-\off);
        \fill[bargreen, draw=magenta!80!gray, line width=1.5pt]   (2, 5.7-\off) rectangle (3, 5.4-\off);
        \fill[barred, draw=black, line width=1.5pt]    (0, 5-\off) rectangle (2, 4.7-\off);
        \fill[barblue, draw=black, line width=1.5pt]    (2, 5-\off) rectangle (4, 4.7-\off);
        \fill[bargreen, draw=black, line width=1.5pt]    (4, 5-\off) rectangle (4.5, 4.7-\off);

        \node[left] at (0, 3.9-\off)  {\small $j_4$};
        \node[left] at (0, 3.2-\off)  {\small $j_5$};
        
        \draw[thick] (0.0, 3.9-\off) -- (4, 3.9-\off); 
        \draw[thick] (0.0, 3.2-\off) -- (5, 3.2-\off); 
        
        \def\off{4.75}
        \node[left] at (0, 2.7-\off)  {\small $j_6$};

        \draw[thick] (2.1, 2.5-\off) -- (6.6, 2.5-\off); 
        \fill[barred, draw=orange!80!gray, line width=1.5pt]    (2.1, 2.9-\off) rectangle (4, 2.6-\off);
        \fill[barblue, draw=orange!80!gray, line width=1.5pt]    (4, 2.9-\off) rectangle (6, 2.6-\off);
        \fill[bargreen, draw=orange!80!gray, line width=1.5pt]   (6, 2.9-\off) rectangle (6.6, 2.6-\off);
         
         \draw [thick, decorate, decoration={ brace, amplitude=5pt}, draw=gray!80!black] 
    (2.1, 3-\off) -- (4.0, 3-\off) 
    node [midway, above=4pt, text=gray!80!black] {$v$};

               \draw [thick, decorate, decoration={brace, amplitude=5pt}, draw=gray!80!black] 
    (4, 3-\off) -- (6.0, 3-\off) 
    node [midway, above=4pt, text=gray!80!black] {$u$};

           \draw [thick, decorate, decoration={ brace, amplitude=5pt}, draw=gray!80!black] 
    (6, 3-\off) -- (6.6, 3-\off) 
    node [midway, above=4pt, text=gray!80!black] {$r$};
        
        \node[left] at (0, 1.8-\off)  {\small $j_7$};
        \node[left] at (0, 1.1-\off)  {\small $j_8$};
        \node[left] at (0, 0.4-\off)  {\small $j_9$};
        
        \draw[thick] (0.0, 1.8-\off) -- (6, 1.8-\off); 
        \draw[thick] (0.0, 1.1-\off) -- (7, 1.1-\off); 
        \draw[thick] (4.1, .4-\off) -- (8.6, .4-\off); 

        
        \node[left] at (0,-0.3-\off)  {\small $j_{10}$};
        \node[left] at (0,-1.0-\off)  {\small $j_{11}$};
        \node[left] at (0,-1.7-\off)  {\small $j_{12}$};
        
        \draw[thick] (0.0, -.3-\off) -- (8, -.3-\off); 
        \draw[thick] (0.0, -1-\off) -- (9, -1-\off); 
        \draw[thick] (6.1, -1.7-\off) -- (10.6, -1.7-\off); 
        
        

        \draw[->, thick] (0,-7) -- (12,-7) node[right] {\textbf{Time}};

            \draw[thick] (0, -7.1) -- (0, -6.9);
            
            \node[below] at (0, -7.1) {$0$};
        \foreach \x/\t in {2/1, 4/2, 6/3, 8/4} {
            
            \draw[thick] (\x, -7.1) -- (\x, -6.9);
            
            \node[below] at (\x, -7.1) {$t_{\t}$};

            \draw[dashed, gray] (\x,-7) -- (\x, 6-\off);

            \ifnum\x>0          
            
            \node[root_small] (r) at (\x, -8) {r};
            \node[mid_small] (u) at (\x+.4, -8.5) {u}; 
            \node[leaf_small] (v) at (\x+.8,-9) {v}; 
            \node[tiny_circle] (u1) at (\x,-9) {$u_{\t}$};
            \node[tiny_circle] (r1) at (\x-.4,-8.5) {$r_{\t}$};
            
            \draw[edge] (u) -- (u1);     
            \draw[edge] (r) -- (r1); 
            \draw[edge] (r) -- (u);     
            \draw[edge] (u) -- (v);  
            
            \fi
            
        }
        
        \node[below left] at (.7,-3-\off) {Alg sends:};

        \node[below left, brown!70!black] at (2.35,-9.5) {$T_1$};
        \node[below left, magenta!80!gray] at (4.35,-9.5) {$T_2$};
        \node[below left, black] at (6.35,-9.5) {$T_3$};
        \node[below left, orange!80!gray] at (8.35,-9.5) {$T_4$};

        \coordinate (target_rect_1) at (6.6, 2.75-\off);
        \coordinate (target_rect_2_in) at (7, 1.05-\off);
        \coordinate (target_rect_2_out) at (7, 1.15-\off);
        \coordinate (target_rect_3_in) at (8, -.35-\off);
        \coordinate (target_rect_3_out) at (8, -.25-\off);
        
        \draw[->, >=stealth, ultra thick, dashed, draw=orange!80!gray] 
        (8.35, -9.5) to [out=0, in=0, looseness=2.5] 
        (target_rect_3_in);
 \draw[->, >=stealth, ultra thick, dashed, draw=orange!80!gray] 
        (target_rect_3_out) to [out=0, in=0, looseness=2.5] 
        node[midway, right, text=orange!80!gray, font=\small\sffamily] {Sequence of CHARGE calls} 
        (target_rect_2_in);
         \draw[->, >=stealth, ultra thick, dashed, draw=orange!80!gray] 
        (target_rect_2_out) to [out=0, in=0, looseness=1] 
        (target_rect_1);

\end{tikzpicture}

%% file: delays_v4.tex
\section{MLAP with Delays}
\label{sec:delays}

We now develop an algorithm and analysis for \mlapdelay that achieves a competitive ratio of $2D$ on trees of depth $D$. 

At a high level, we say that a set of requests is \emph{critical} when their accumulated delay is sufficient to pay for the cost of serving them. Just as in the deadline setting, the algorithm only schedules a service when some request set becomes critical. At this time, we then identify other pending requests that can be added to this service at low additional cost.

Here a key problem arises: while in the deadlines setting, there was a clear priority order over requests in order of increasing deadlines, in the delay setting it is unclear which requests should be prioritized over others. We take inspiration from the dual LP to guide this process. In particular, we construct a ``fake'' dual solution by raising duals of the pending requests gradually and simultaneously, matching the delay accumulated by each request. A set of requests can be added to the scheduled service if its duals can pay for the cost of the subtree that needs to be added to cover the requests. The requests are then prioritized in increasing order of the time it takes them to cover their service cost in this process. 

It would be convenient to charge the cost of each tree scheduled to the duals that ``sponsor'' that tree. Unfortunately, the fake duals described above are not feasible for the dual LP. In particular, the fake dual assigned to a request in a future iteration may interfere with the feasibility of a dual constraint at the time of a previous iteration. In order to perform our analysis, we define a second dual solution {\em in hindsight} after all requests have been received and processed. This hindsight dual is feasible and, in general, smaller than the fake dual used by the algorithm. Nevertheless, it can pay for a $1/D$ fraction of the cost of the primal solution.

\subsection{Algorithm}

We begin with a few  definitions. We do not explicitly define a fake dual solution, but rather track the accumulated delay of any request set over time to determine when it can be serviced. 

\begin{definition}[Saturation and Reachtime]
\label{def:saturation}
    Let $u \in \mc{T}$ be an arbitrary node and $Q$ be a set of requests such that $Q_u \neq \emptyset$. We then define the following terms.

    \begin{itemize}
        \item {\bf Saturation time:} $\sattime(Q, u) = \min\{ t \mid \exists Q' \subseteq Q_u \text{ such that } \del(Q', [t]) = w(\tree_u^{Q'})\}$ is the earliest time $t$ when the accumulated delay of some requests in $Q_u$ equals the cost of connecting them to $u$. We say that $Q$ \emph{saturates}  $u$ at time $\sattime(Q,u)$.
        
        \item {\bf Critical set and subtree:} 
        Suppose $Q$ saturates the root node $r$ at time $t$. The \emph{critical set} is the specific subset $Q' \subseteq Q$ that attains this saturation time, i.e., $\del(Q', [t]) = w(\tree_r^{Q'})$. The subtree $T^{Q'}_r$ is called the \emph{critical subtree}.
        
        \item {\bf Reach time and reach set:} $\reachtime(Q, u) \coloneq \min_{v \in \children(u)} \sattime(Q, v)$ is the earliest time at which some requests in $Q$ saturate a child of $u$. 
        
        The specific subset $Q' \subseteq Q_v$ that attains the saturation time for the minimizing child $v$ is called the $\reachset$ of $(Q, u)$. In case of ties, we choose a maximal such subset $Q'$.

    \end{itemize}
\end{definition}

We note that saturation times and reach times depend solely on the subset of requests being considered and are independent of the decisions made by the algorithm. This is an important property that we will use in our analysis.

We are now ready to describe the algorithm. In every iteration, as new requests arrive over time and all pending requests accumulate delay, the algorithm waits until a set of requests becomes critical at time $t$. We break any ties in favor of the maximal such critical set. Let $S$ be the subtree spanned by the root and the leaves of the critical request set. The algorithm eventually schedules a superset of this {\em critical subtree}.

Just as in the deadlines setting, we initialize $T$ to be the critical subtree and run a ``simulation'' from every node in $T$ in a \emph{leaf-to-root} order. The simulation at node $u$ is endowed with a budget of $w(u)$ that can be used towards adding new nodes to the tree. During the simulation step, as long as there are pending requests inside $\tree_u$ and the budget is not exhausted, we consider sets of requests in order of increasing \emph{reach times} to $u$ and invest in the corresponding nodes. Once any node has been fully paid for it gets added to the tree $T$. We also run simulations from each node newly added to the tree, again in leaf-to-root order.
Once all simulations terminate, the constructed tree $T$ is scheduled and a new iteration begins. A formal description is given in Algorithm \ref{alg:multilevel_agg_delays}.

\begin{algorithm}[htbp]
\caption{Online Multilevel Aggregation with Arbitrary Delay Functions}
\label{alg:multilevel_agg_delays}
\begin{minipage}{\linewidth}
\begin{algorithmic}[1]
\State \textbf{Initialization:}
\State Initialize $b_v, c_v \leftarrow 0$ for every node $v \in \mc{T}$
\Statex
\Function{\UponCritical()}{} \Comment{Wait until some request set becomes critical} 
    \LetState{$S$} {critical subtree}
    \LetState{$T$} {$\emptyset$}
    \State \Call{\Add}{$S, T$} \label{step:add-critical}
    \State transmit $T$ 
\EndFunction
\Statex
\Function{\Add}{$V, T$}     
    \State Set $c_v \gets 0,\ \ \forall v \in V$ 
    \State Add all nodes in $V$ to tree $T$
    \State Mark all pending requests on $v \in V$ as served
    \For{$v \in V$ in leaf-to-root order}
                \State \Call{\Simulate}{$v, T$}
    \EndFor
\EndFunction
\Statex
\Function{\Simulate}{$u, T$}
    \State $b_u \gets w(u)$
    \While{$b_u \neq 0$ \textbf{and} there remain pending requests in $\tree_u$}
        \LetState{$Q$}{set of pending requests in $\tree_u$} 
        \LetState{$Q'$}{$ \reachset(Q,u)$}\footnote{Note that we must saturate the child of $u$ even if the child of $u$ is already in the tree. This subtlety is important for proving Lemmas \ref{lem:freezetime-dominance} and \ref{lem:bought-nodes-freeze-earlier}. Also note that computing the reach set of $u$ is polynomial in the number of vertices in $\mathcal{T}_u$ and the number of pending jobs in $\mathcal{T}_u$. One way to efficiently compute the reach set is to look into the future and keep track of the set of nodes in $\mathcal{T}_u$ saturated by the pending requests in $\mathcal{T}_u$. Each time a new node is saturated, we can add that node to the set of saturated nodes, and stop once a child of $u$ has been saturated.
} and $v \gets$ corresponding child node of $u$ \label{line:witness} 
        \LetState{$T_{Q'}$}{$\tree_v^{Q'} \setminus T$} 
        \LetState{$\cst(T_{Q'})$}{$\sum_{v' \in T_{Q'}} (w(v') - c_{v'})$}
        \If{$\cst(T_{Q'}) \leq b_u$} 
            \For{$v' \in T_{Q'}$}
                \State \Call{\Invest}{$u, v', w(v') - c_{v'}$} \label{line:invest1}
            \EndFor
            \State $T_{part} \gets T$ \Comment{defined only for ease of analysis}
            \State \Call{\Add}{$T_{Q'}, T$} \label{step:add-noncritical} 
        \Else
            \For{$v' \in T_{Q'}$}
                \State $y(u, v') \gets b_u \cdot \frac{w(v') - c_{v'}}{\cst(T_{Q'})}$ 
            \EndFor
            \For{$v' \in T_{Q'}$}
                \State \Call{\Invest}{$u, v', y(u, v')$} \label{line:invest2} 
            \EndFor
        \EndIf
    \EndWhile 
\EndFunction
\Statex
\Function{\Invest}{$u, v, \gamma$} 
    \State increase $c_{v}$ by $\gamma$
    \State decrease $b_u$ by $\gamma$
\EndFunction

\end{algorithmic}
\end{minipage}
\end{algorithm}

Although our high-level approach resembles \cite{azar2019general}, we highlight several key distinctions. First, we always schedule a superset of the critical request set, whereas \cite{azar2019general} may only schedule a subset. Second, our simulation invests in all nodes required to serve a particular request set simultaneously rather than sequentially. Finally, simulations strictly follow a leaf-to-root order, which is crucial to proving one of our key technical lemmas (Lemma $\ref{lem:freezetime-dominance}$). See the proofs of Lemmas \ref{lem:all-or-none} and \ref{lem:bought-nodes-freeze-earlier} for more details. These modifications are key to enabling our dual fitting analysis.

Our goal in this section is to establish Theorem \ref{thm:delays-main}, which follows directly from the upper and lower bounds established in Corollary \ref{cor:primal_main} and Lemma \ref{lem:opt-bound-delays}, for Algorithm \ref{alg:multilevel_agg_delays} and the optimal solution, respectively. 

\subsection{Primal Analysis}

In the following, we bound the total delay cost incurred by the algorithm by the cost of the trees transmitted by the algorithm. The following simple lemma follows from the fact that the algorithm serves a request set as soon as it becomes critical. We defer the formal proof to Section \ref{sec:primal-deferred}.

\begin{restatable}{lemma}{algbounddelaystwo}
\label{lem:alg-bound-delays2}
Let $(T_1, t_1), \dots, (T_k, t_k)$ be the solution generated by Algorithm~\ref{alg:multilevel_agg_delays}. Then, 
    \[\delayalg \leq \sum_{i=1}^k w(T_i).\]
\end{restatable}

As in the deadlines section, the cost of the trees transmitted by the algorithm can be related to the cost of the critical subtrees. We again let $c_v(t)$ denote the unspent investment received by node $v$ up to time $t$ 
and define $\hat{w}(v, t) \coloneq w(v) - c_v(t) \geq 0$ as the unpaid cost of $v$ at time $t$. The following lemma states that the total cost of all transmitted trees is at most $D$ times the unpaid costs of nodes in the critical subtrees. Intuitively, when some node $u$ is included in the critical subtree, it is endowed with a budget of $w(u)$ which it uses to buy nodes inside the subtree rooted at $u$ of total cost $w(u)$. These nodes then further buy more nodes that are successively deeper in the tree structure. Thus in the worst case, a node $u$ can be ultimately responsible for adding nodes of total weight $D \cdot w(u)$ to the tree. The formal proof is a bit more involved and uses a preflow argument to properly account for investments that only partially buy nodes and is deferred to Section \ref{sec:primal-deferred}.

\begin{restatable}{lemma}{algbounddelays}
\label{lem:alg-bound-delays}
Let $(T_1, t_1), \dots, (T_k, t_k)$ be the solution generated by Algorithm~\ref{alg:multilevel_agg_delays}, where $t_i$ is the time just before tree $T_i$ is transmitted. Let $S_i \subseteq T_i$ be the critical subtree of tree $T_i$. Then,
    \[\treecostalg = \sum_{i=1}^k w(T_i) \leq D \cdot \sum_{i=1}^k \sum_{v \in S_i} \hat{w}(v, t_i).\]
\end{restatable}

We obtain the following corollary.
\begin{corollary}
\label{cor:primal_main}
    $\costalg \leq 2D\cdot \sum_{i=1}^k \sum_{v \in S_i} \hat{w}(v, t_i).$
\end{corollary}

\subsection{Dual Fitting}

The following lemma, a generalization of Lemma~\ref{lem:deadlines-opt-lb} to the delay setting, provides a lower bound on the cost of an optimal solution and together with Corollary \ref{cor:primal_main} implies a competitive ratio of $2D$.

\begin{lemma}
\label{lem:opt-bound-delays}
Let $(T_1, t_1), \dots, (T_k, t_k)$ be the solution generated by Algorithm~\ref{alg:multilevel_agg_delays}, and $S_i \subset T_i$ be the critical subtree of tree $T_i$. Then,
$$\costopt \geq \sum_{i=1}^k \sum_{v \in S_i} \hat{w}(v, t_i).$$
\end{lemma}

The proof of Lemma \ref{lem:opt-bound-delays} is the core technical contribution of the paper and relies on the construction of a feasible dual solution. The dual construction procedure loosely mirrors the one in the deadlines setting but is significantly more nuanced. We now describe the differences. 

First, in the delay setting the dual variable $\alpha_{jt}$ is restricted to be at most $\del(j, t)$. This differs from the deadline setting, in which there is no restriction on how large $\alpha_j$ can be, aside from the feasibility constraint imposed by each rooted subtree.

Second, recall that in the deadlines setting, the cost of each vertex $v$ in the critical path is paid for by either a single request (namely the critical request that caused the vertex to be included in the service) or recovered recursively from vertices that $v$ invests in. In the delays setting, we will define a similar Charge procedure, but each vertex will charge a set of requests that ``sponsored'' its inclusion in the critical subtree. We will call this set the {\em patrons} of $v$. 

We will use the following terminology. For a tree $T$ sent by the algorithm and vertex $v$ in $T$, the set $\patrons(v,T)$ is a subset of requests in $T_v$ that pay for the cost of $v$. For each patron $j\in \patrons(v,T)$, we use $I(j,v,T)$ to denote the interval of time over which $v$ charges the dual variables corresponding to request $j$. The total amount of dual charged by $v$ to patron $j\in \patrons(v,T)$ is denoted by $\lambda(j,v,T)$. 

The intervals over which requests are charged may extend beyond the time at which the requests are served. We keep track of the latest time that a request $j$ can be charged through a quantity called its {\em freezetime} and denoted $\freezetime(j)\ge a_j$. We extend this notion to vertices. Intuitively, for a subtree $T$ served by the algorithm and some vertex $v\in T$, $\freezetime(v, T)$ denotes the latest time at which $v$ charges some dual variable.

We defer the formal description of these quantities to Algorithm \ref{alg:earliest_arrival_first_assignment} in Section~\ref{sec:patrons}, but note that they satisfy the properties formalized in the following two lemmas.

\begin{restatable}{lemma}{patronseqcost}
\label{lem:patrons_equal_cost}
Let $T$ be a tree transmitted by Algorithm \ref{alg:multilevel_agg_delays}. For any node $v$ in $T$ and request $j$ that is served by $T$, let $\patrons(v, T)$, intervals $I(j, v, T)$, and charges $\lambda(j, v, T)$ be as defined by Algorithm \ref{alg:earliest_arrival_first_assignment}. Then,
\begin{enumerate}[label=(\roman*)]
    \item $I(j, v_1, T) \cap I(j, v_2, T) = \emptyset$ for any request $j$ and distinct nodes $v_1, v_2$ in $T$.
    \item For any request $j$ served in $T$, $\cup_{v \in T} I(j,v,T) \subseteq [a_j, \freezetime(j)]$. 
    \item For all $v \in T$ and $j\in \patrons(v,T)$, we have $\lambda(j,v,T) = \delay(j, I(j,v,T))$, i.e., the total amount charged to patron $j$ equals the delay accumulated by $j$ during the interval $I(j,v,T)$. 
    \item For all $v$ and $T$, this charging scheme completely recovers the cost of $v$.\\ That is, $\sum_{j\in \patrons(v,T)} \lambda(j,v,T) = w(v)$.
\end{enumerate}
\end{restatable}

For a tree $T$ scheduled by the algorithm, we characterize each node $v\in T$ as either {\em disjoint} or {\em overlapping} depending on whether or not the patrons of $v$ were pending at the previous time $v$ was scheduled.
\begin{definition}[Disjoint and Overlapping Nodes]
    Consider a node $v$ sent in tree $T$ at time $t$. Let $t_{\prev}(v) < t$ be the previous time when $v$ was sent. Then node $v$ is said to be \emph{disjoint in tree $T$} if $\forall \: j \in \patrons(v, T), a_j > t_{\prev}(v)$. Otherwise, $v$ is said to be \emph{overlapping in tree $T$}.
\end{definition}
We pay for disjoint and overlapping nodes in different ways in our dual construction. Disjoint nodes charge their patrons directly according to the intervals $I(j,v,T)$ described above.
By the definition of disjoint nodes, the duals of these patrons do not interfere with older dual constraints. Overlapping nodes, on the other hand, cannot directly charge their patrons. Instead, we pass the charge for an overlapping node down the tree to the descendants the node invested in the previous time the node was scheduled. The charge passed down by a vertex $u$ to a node $v$ it invested in is proportional to the fraction of its weight $u$ invested in $v$. 

We now state a crucial property about freezetimes that we call the \emph{Freezetime Dominance Lemma}. This property is the technical crux behind the dual feasibility for the delay setting and is the analog of Claim \ref{claim:ddl_v_less_than_ddl_u} in the deadlines setting. The proof of Lemma \ref{lem:freezetime-dominance} is fairly technical and deferred to a separate section (Section \ref{sec:freezetime_analysis}). 

\newif\iforiginalstate
\originalstatetrue
\newcommand{\onefootnote}[1]{\iforiginalstate\footnote{#1}\fi}

\begin{restatable}{lemma}{freezetimelem}{\em [Freezetime Dominance]}
\label{lem:freezetime-dominance}
    Suppose $u$ invests in node $v$ in tree $T$ at time $t$, and $v$ is next sent in tree $T^{v}$ at time $t^v \geq t$\onefootnote{Note that we may have $T^v = T$ if $v$ gets scheduled in the same tree either during $u$'s simulation or because of some other node's simulation.}. Further suppose that after $T$, $u$ is next sent in tree $T_{\nxt}$ at $t_{\nxt} > t$ and $u$ is overlapping in tree $T_{\nxt}$. Then
    \[
        \freezetime(v, T^v) \leq \freezetime(u, T_{\nxt}).
    \]
\end{restatable}

The dual construction procedure is described in pseudocode in Algorithm \ref{alg:delays-analysis}.

\begin{algorithm}[H]
\caption{Dual Construction Algorithm}
\label{alg:delays-analysis}
\begin{algorithmic}[1]
\State Initialize $\alpha_{jt} \leftarrow 0$ for every $j, t$
\Statex
\For{$i = 1, 2, \ldots, k$}
\LetState{$T_i$}{tree transmitted at time $t_i$}
\LetState{$S_i$}{critical subtree in $T_i$}
\For{$u \in S_i$}
\State \Call{\Charge}{$u$, $T_i$, $\hat{w}(u, t_i)$} \Comment{{\footnotesize $u$ charges its unpaid cost in tree $T_i$ to a subset of requests}} \label{line:main-charge-call}
\EndFor
\EndFor
\Statex
\Function{\Charge}{$u$, $T$, $\gamma$} \Comment{{\footnotesize $u$ is transmitted in tree $T$, and charges $\gamma$ to a subset of requests}}   
    \LetState{$t$}{time that $T$ was transmitted} 
    \LetState{$t_{\prev}$}{previous time strictly before $t$ when $u$ was transmitted, or $0$ otherwise}  
    \LetState{$T_{\prev}$}{tree sent at $t_{\prev}$, if one exists}
    \Statex
    
    \If{$u$ is disjoint in tree $T$}
        \For{$j\in \patrons(u,T)$}
            \State Set $\alpha_{j\tau} = \delay(j, \tau) \cdot \frac{\gamma}{w(u)}$ for all $\tau\in I(j,u,T)$ 
            \LetState{$\gamma_j$}{$ \delay(j, I(j,u,T)) \cdot \frac{\gamma}{w(u)}$}
            \Statex \Comment{{\footnotesize $u$ charges every patron $j$ in proportion to $\delay(j,I(j,u,T))$}}
        \EndFor
        \State {\Return{$\bigcup_{j \in \patrons(u, T)} \{(j, \gamma_j)\}$}}
    \Else
        \LetState{$A_u$}{Set of nodes that $u$ invests in during $T_{\prev}$}
        \Statex \Comment{{\footnotesize $u$ will recursively charge requests through proxy nodes in $A_u$}}
        \For{$v \in A_u$}
            \LetState{$y(u, v, T_{\prev})$}{Amount invested by $u$ in $v$ in tree $T_{\prev}$}
            \LetState{$T_{\nxt}^v$}{next tree where $v$ is transmitted (inclusive of $T_{\prev}$)}
            \label{line:tnext}
            \State {$B_v \gets $} \Call{\Charge}{$v$, $T_{\nxt}^v$, $y(u, v, T_{\prev}) \cdot \frac{\gamma}{w(u)}$} \Statex \Comment{{\footnotesize The amount charged by $v$ is proportional to the investment $u$ makes in $v$}}
        \EndFor
        \State {\Return{$\bigcup_{v \in A_u} B_v$}}
    \EndIf
\EndFunction
\end{algorithmic}
\end{algorithm}

\subsubsection{Dual Analysis}

Let $\{\alpha\}$ denote the dual solution returned by Algorithm~\ref{alg:delays-analysis}. In order to prove Lemma \ref{lem:opt-bound-delays}, we need to show that the objective of the dual solution is at least the unpaid cost of all the critical subtrees sent by the algorithm and that $\{\alpha\}$ is a feasible dual solution. We provide a high level overview of both results here and defer all missing proofs to Section \ref{sec:deferred-proofs-delays}.

\begin{restatable}{claim}{delaysdualobjinc}
\label{claim:delays_dual_obj_inc}
Each call to $\Charge(u, T, \gamma)$ increases the dual objective $\sum_{j,t} \alpha_{jt}$ by $\gamma$. 
\end{restatable}
\begin{proof}[Proof Sketch]
    The claim follows by induction over the tree of recursive $\Charge(\cdot, \cdot, \cdot)$ calls. If $u$ is disjoint in $T$, then the total delay associated with its patron intervals is exactly $w(u)$ (Lemma~\ref{lem:patrons_equal_cost}). Since the dual increases by $\sum_{j \in \patrons(u, T)} \del(j, I(j,u,T)) \cdot \frac{\gamma}{w(u)} = \gamma$, this completes the base case of the recursion. Otherwise if $u$ is overlapping, then $u$ must have exhausted its budget at time $t_{\prev}$. By the induction hypothesis, the recursive calls together increase the dual by $\sum_{v \in A_u} y(u, v, T_{\prev}) \cdot \frac{\gamma}{w(u)} = \gamma$.
\end{proof}

The following corollary is then immediate. 

\begin{corollary}
    $\sum_{j,t} \alpha_{jt} = \sum_{i=1}^k \sum_{v \in S_i} \hat{w}(v, t_i)$.
\end{corollary}

\paragraph{Feasibility of the dual.} The rest of the section is devoted to proving that the constructed dual solution is feasible, which immediately implies Lemma~\ref{lem:opt-bound-delays}. 

Recall that each call to
$\Charge(u, T, \gamma)$ returns a set~$B$ of
(request, amount) pairs: each pair $(j, \gamma') \in B$
records that $\alpha_j$ was increased by $\gamma'$, either
directly in this call or through a recursive sub-call. Just as in the deadlines setting, we utilize the attributed payments $\delta_{vtj}$ to argue about feasibility of the constructed dual solution. These payments are defined as follows. 

\begin{definition}[Assignment variables]
\label{def:delta-delay}
All variables $\delta_{u\tau j}$ and $\delta_{u\tau}$ are initialized to $0$.
For each call $\Charge(u, T, \gamma)$ executed during
Algorithm~\ref{alg:delays-analysis}, let $t$ be the
transmission time of~$T$, let $t_{\prev}$ be the last
time~$u$ was transmitted before~$t$, and let $B$ denote
the set of pairs returned by this call. Then for each $\tau \in (t_{\prev}, t]$, we increment $\delta_{u \tau} \leftarrow \delta_{u \tau} + \gamma$ and for each
$(j, \gamma') \in B$ and each $\tau \in (t_{\prev}, t]$,
we increment
\[
  \delta_{u\tau j} \;\leftarrow\; \delta_{u\tau j} + \gamma'.
\]
\end{definition}

We can now argue feasibility. Lemma \ref{lem:delays_feas_deltas_alphas} says that every $\alpha_j$ is being attributed to some node $v$ on its path for the entire time it is alive via $\delta_{vtj}$. This implies that it suffices to reason about the feasibility of the $\delta_{vtj}$s. Lemma \ref{lem:delays_feas_deltas_deltas} states that at every node $v$ and time instant $t$, the sum of $\delta_{vtj}$ for all $j$ is exactly equal to $\delta_{vt}$. Therefore, it suffices to reason about the feasibility of the $\delta_{vt}s$. Lemma \ref{lem:delays_feas_deltas_weights} states that all $\delta_{vt}s$ are feasible (less than $w(v)$) at every time step $t$. Finally, Lemma \ref{lem:delays_config_feasibility} shows that the above three lemmas are sufficient to prove feasibility for the dual of the configuration LP.  

\begin{restatable}{lemma}{delays_delta_cover_alpha}
\label{lem:delays_feas_deltas_alphas}
For every $j \in \mathcal{J}$ and $t \in [a_j, \freezetime(j)]$, $\sum_{v \in P(j)} \delta_{vtj} \ge \int_{t}^{\infty} \alpha_{j\tau}d\tau$.
\end{restatable}

\begin{proof}
Fix a request $j$. We will prove the claim for all $t\in [a_j, \freezetime(j)]$ by induction over the course of Algorithm~\ref{alg:delays-analysis}. At the beginning of the algorithm, all $\alpha$s and $\delta$s are zero, so the claim holds.

Consider an iteration where the algorithm increases some $\alpha_{j\tau}$ value. Let $\Charge(u_1, T_1, \gamma_1)$ be the innermost call to $\Charge$ that makes this increase. Observe that $u_1$ is disjoint in $T_1$. We will now consider the sequence of nested charge calls that lead to the call $\Charge(u_1, T_1, \gamma_1)$. Let this sequence be of length $k$. The $i$th call in this sequence is denoted $\Charge(u_i, T_i, \gamma_i)$, and for $i>1$, $\Charge(u_i, T_i, \gamma_i)$ calls $\Charge(u_{i-1}, T_{i-1}, \gamma_{i-1})$. 

Let $T_{prev, i}$ denote the last tree strictly before $T_i$ where $u_i$ was transmitted. Let $t_i$ denote the time at which $T_i$ is sent, and $t_{prev, i}$ the time at which $T_{prev, i}$ is sent. 
Observe that $\Charge(u_1, T_1, \gamma_1)$ increases $\alpha_{j \tau}$ over the interval $I(j, u_1, T_1)$ for a total amount of $\gamma_1$. In the sequence of $\Charge$ calls defined above, for each $u_i$, $i\in [k]$, we increase $\delta_{u_i\tau j}$ 
by the same $\gamma_1$ amount over the interval $(t_{prev, i}, t_i]$. In order to prove the lemma, we will argue that these intervals together cover the entire interval $[a_j, \freezetime(j)]$.

In order to see this, we first observe that since $u_1$ is disjoint in $T_1$, $a_j>t_{prev, 1}$. Furthermore, for $i>1$, $t_{prev, i}\le t_{i-1}$ by the definition of $T_{i-1}$ in Line~\ref{line:tnext}. 

To conclude the proof we will argue that $\freezetime(j)\le t_k$. First, note that $u_k$ is critical in $T_k$, so $\freezetime(u_k, T_k) = t_k$. Next, recall that for $i\in \{2, \cdots, k\}$, we have that $u_i$ is overlapping in $T_i$. Moreover, $u_i$ invested in $u_{i-1}$ in $T_{prev, i}$. Then, we can apply Lemma~\ref{lem:freezetime-dominance} to obtain:
\[\freezetime(u_{i-1}, T_{i-1})\le\freezetime(u_i, T_i).\]
Recalling that $\freezetime(j)= \freezetime(u_1, T_1)$ and $\freezetime(u_k, T_k)=t_k$, and combining these inequalities implies the claim.
\end{proof}

The following lemma follows from simple bookkeeping.

\begin{lemma}\label{lem:delays_feas_deltas_deltas}
For all vertices $v$ and times $t$, $\sum_{j \in \mathcal{J}} \delta_{vtj} = \delta_{vt}$.
\end{lemma}

\begin{proof}
We first claim that for all calls to $\Charge$ in Algorithm~\ref{alg:delays-analysis}, with input $(u,T,\gamma)$ and output $B$, we have $\gamma = \sum_{(j,\gamma_j)\in B} \gamma_j$. Using this claim, we will prove the statement of the lemma by induction over the course of Algorithm~\ref{alg:delays-analysis}. At the beginning of the algorithm, all $\delta$s are zero, so the base case holds.

Consider a $\Charge(u, T, \gamma)$ call. Observe that all of the $\delta_{u\tau}$ and $\delta_{u\tau j}$ variables that increase during this call (outside of recursive calls) only increase over the interval $(t_{\prev}, t]$ for $t_{\prev}$ and $t$ as defined in this iteration of $\Charge$. So, it suffices to prove that the increase in $\delta_{u\tau}$ is equal to the increase in $\sum_j\delta_{u\tau j}$ for $\tau\in (t_{\prev},t]$ in this iteration. But because $\delta_{u\tau}$ increases by $\gamma$ and $\delta_{u\tau j}$ increases by $\gamma_j$ for $j\in B$, the statement follows immediately from our claim that $\gamma = \sum_{(j,\gamma_j)\in B} \gamma_j$.

We will now prove the claim by induction over the recursive structure of the calls to $\Charge$. Consider such a call $\Charge(u,T,\gamma)$. Suppose that $u$ is disjoint in $T$. Then, we have 
\[\sum_{(j,\gamma_j)\in B} \gamma_j = \sum_{j\in \patrons(u,T)} \lambda(j,u,T)\cdot \frac{\gamma}{w(u)} = \gamma \cdot \frac{\sum_{j\in\patrons(u,T)} \lambda(j,u,T)}{w(u)} = \gamma\]
where the last equality follows from Lemma~\ref{lem:patrons_equal_cost}.

Suppose that $u$ is overlapping in $T$. Then, we have,
\[ \sum_{(j,\gamma_j)\in B} \gamma_j = \sum_{v\in A_u} \sum_{(j,\gamma_j)\in B_v} \gamma_j = \sum_{v\in A_u} y(u,v,T_{\prev}) \cdot \frac{\gamma}{w(u)} = \gamma\]
where the second equality follows from the induction hypothesis, and in the last equality we used the fact that the total amount invested by $u$ in $T_{\prev}$, $\sum_{v\in A_u} y(u,v,T_{\prev})$, is exactly equal to $w(u)$.
\end{proof}

Lemma~\ref{lem:delays_feas_deltas_weights} is also a simple bookkeeping argument, but requires a careful accounting of the possibly multiple calls to $\Charge(v,T,\cdot)$ for a single pair $(v,T)$. The following simple claim tallies the total charge corresponding to such a pair and we defer the proof to Section \ref{sec:deferred-proofs-delays}.

\begin{restatable}{claim}{delayslimitedcharges}
\label{claim:delays_limited_charges}
Fix a vertex $v$ and $T$. Suppose \Charge($v$, $T$, $\gamma_i$) is called $\kappa > 0$ times. Then $\sum_{i \in [\kappa]} \gamma_i \leq w(v)$.
\end{restatable}

\begin{lemma}\label{lem:delays_feas_deltas_weights}
For all vertices $v$ and times $t$, $\delta_{vt} \leq w(v)$.
\end{lemma}

\begin{proof}
    Fix a $v$ and $t$ with $\delta_{vt}>0$. Define $t_1$ as the latest time strictly before $t$ when $v$ was transmitted, and $0$ if no such time exists. Define $t_2$ as the earliest time $\geq t$ when $v$ was transmitted; Such a time must exist if $\delta_{vt}>0$. Let $T_{\nxt}^v$ be the tree transmitted at time $t_2$. Observe that $\delta_{vt}$ is increased by $\gamma$ if and only if \Charge($v$, $T_{\nxt}^v$, $\gamma$) is called. The lemma then follows from Claim \ref{claim:delays_limited_charges} by summing $\gamma$ over all calls to \Charge($v$, $T_{\nxt}^v$, $\cdot$).
\end{proof}

\begin{lemma}\label{lem:delays_config_feasibility}
The $\alpha_{jt}$ variables form a feasible dual solution.
\end{lemma}
\begin{proof}
    Consider any subtree $S\in \mathcal{S}$ and time $t$. Let $J(S, t) = \{j | a_j \leq t \text{ and } v_j \in S\}$ be the requests that contribute to the LHS of the dual constraint. 
    We now have the following:
\begin{align*}
    \sum_{j \in J(S, t)} \int_{t}^{\infty} \alpha_{j\tau}d\tau \le \sum_{j \in J(S, t)} \sum_{v \in P(j)} \delta_{vtj} \le \sum_{v \in S}\sum_{j\in \mathcal{J}} \delta_{vtj} \le \sum_{v \in S} \delta_{vt} \leq \sum_{v \in S} w(v) = \cost(S).
\end{align*}
where the first inequality follows from Lemma \ref{lem:delays_feas_deltas_alphas} along with the observation that $\int_{t}^{\infty} \alpha_{j\tau}d\tau = 0$ for all $t > \freezetime(j)$. The second inequality follows by recalling that $P(j)\subseteq S$ for all $j\in J(S,t)$ and $J(S,t) \subseteq \mc{J}$, and changing the order of summation over $v$ and $j$. The third and fourth inequalities follow from Lemmas \ref{lem:delays_feas_deltas_deltas} and \ref{lem:delays_feas_deltas_weights} respectively.
\end{proof}

\input{patrons}

\input{freezetime}

%% file: patrons.tex
\subsection{Patrons and Freezetimes}
\label{sec:patrons}
In this section we formally define freezetimes for vertices and requests, and provide an algorithm for assigning patrons to vertices satisfying Lemma~\ref{lem:patrons_equal_cost}. The proof of Lemma \ref{lem:patrons_equal_cost} can be found in Section \ref{sec:eaf_proof}.

We begin with some terminology.

\begin{definition}
    In our algorithm and analysis, there are several ways in which nodes and requests are associated with each other:
    \begin{itemize}
        \item {\bf Investment:} We say that a node $u$ \emph{invests} in another node $v$ in tree $T$ if $\Invest(u, v, \cdot)$ is called on Lines~\ref{line:invest1} or \ref{line:invest2} within $\Simulate(u, T)$.
        \item {\bf Sponsorship and witness sets:} If $\Invest(u, v, \cdot)$ is called within $\Simulate(u, T)$ due to the set $Q'$ of requests (defined on Line~\ref{line:witness} of the same iteration), we say that $\Simulate(u, T)$ \emph{witnesses} $Q'$ and $Q'$
        \emph{sponsors} $v$. 
    \end{itemize}
\end{definition}
We emphasize that the patrons of a vertex $v$ are not necessarily the same set of requests that sponsor the vertex. 
While sponsorship is determined in an online fashion during the algorithm's execution, patrons are determined in hindsight depending on which requests were eventually served by the algorithm in every iteration. We next define freezetimes.

\begin{definition}[Freezetime]
    Suppose that a vertex $v$ is added to a tree $T$ in Step~\ref{step:add-critical} of the algorithm and $T$ is serviced at time $t$. Then we define $\freezetime(v,T)\coloneq t$.

    Suppose that a vertex $v$ is added to a tree in Step~\ref{step:add-noncritical} of the algorithm. Let $T$ denote the eventual tree scheduled and $T_{part}$ denote the tree just before Step~\ref{step:add-noncritical} is executed. Let $Q'$ denote the sponsor set in this iteration, and $A = P(v)\setminus T_{part}$ denote the set of all ancestors of $v$, including $v$ itself, that are not in $T_{part}$. Then we define $\freezetime(v,T) \coloneq \max_{u\in A} \sattime(Q'_u,u)$. This is the earliest time $Q'$ saturates $v$ and all of its ancestors in $A$.

    Let $j$ be any request and let $T$ be the tree in which it is served. Then we define $\freezetime(j)\coloneq \freezetime(v_j, T)$.
\end{definition}

In Figure \ref{fig:freezetime_example}, we give an example illustrating how freezetime and reachtime are defined.

\begin{figure}[h]
    \centering
        \begin{adjustbox}{width=1\textwidth}
    \input{freeztime_tikz}
    \end{adjustbox}
    \caption{An illustration of a $\Simulate(r, T)$ call at time $0$. The nodes $r$ and $w$ already belong to tree $T$ prior to the call, whereas nodes $u, v, u_1,$ and $u_2$ are newly added to the tree during the execution of $\Simulate$. All nodes have a cost of $1$, except for node $w$ which has a cost of $7$. Each request $j_i$ is represented by a dashed square connected to its corresponding node of origin. Requests $j_1, j_2,$ and $j_3$ accumulate linear delay over time at  rates $2, 0.5$, and $0.25$, respectively. Note that $\reachtime(\{j_1, j_2, j_3\}, r) = 4$, since by time $t=4$, $j_1$ contributes $6$ and $j_2$ contributes $1$ unit of delay toward vertex $w$; Note that the some nodes reach their freezetimes earlier because node $w$ was already present in $T$ before the $\Simulate$ call began.}
    \label{fig:freezetime_example}
\end{figure}
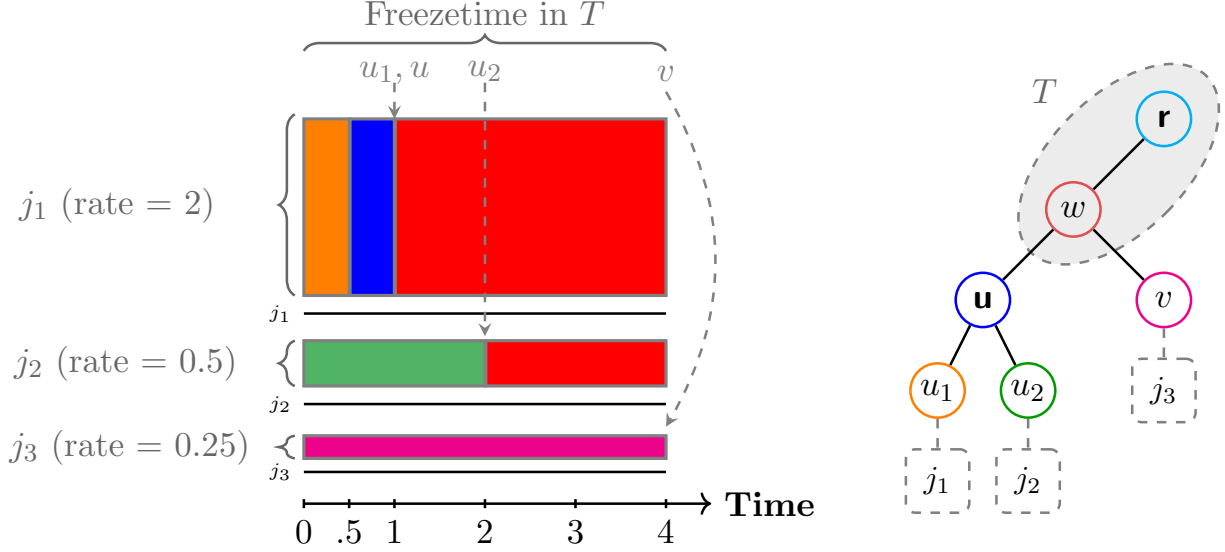

\subsubsection{Construction of Patron Sets}

Intuitively, each node $u$ that is included in a transmitted tree $T$ is added to $T$ because the delay accumulated by some requests was sufficient to pay for the cost of $u$. In the course of the online algorithm, each node is either part of the critical subtree or is sponsored by some request set $Q'$ during a simulation. However, this sponsorship relationship is not convenient for dual construction. Instead, in this section, we define a different association between nodes and requests. For each node $u$ in a transmitted tree $T$, we associate a set of patron requests. For each request $j$ that is a patron of node $v$ (in tree $T$), we also associate a time interval $I(j, v, T) = [t_1, t_2] \subseteq [a_j, \freezetime(j)]$ to maintain that $\sum_{j \in \patrons(v, T)} \del(j, I(j, v, T)) \geq w(v)$. Additionally, the intervals $\{I(j, v, T)\}_{v, T}$ are disjoint for any request $j$.

Algorithm \ref{alg:earliest_arrival_first_assignment} gives a detailed pseudocode for construction of patron sets of a node and their associated intervals. Informally, suppose the delay accumulated by a set of requests $Q$ is sufficient to cover the cost of a set $V$ of nodes. Then we consider requests sequentially in order of earliest arrival first. For each request $j$, we then associate the delay incurred by the request from its arrival time $a_j$ to its freeze time with nodes in $V$ in a leaf-to-root order.

Observe that by construction for all $v,T$ and $j\in\patrons(v,T)$, we have $\freezetime(j)=\freezetime(v,T)$. Furthermore, properties (i)-(iii) in Lemma \ref{lem:patrons_equal_cost} are satisfied by construction. We prove property (iv) of the lemma in Section \ref{sec:eaf_proof}, showing that Algorithm \ref{alg:earliest_arrival_first_assignment} terminates with each node $v$ fully paid for by the delay incurred by its patrons. 

\begin{algorithm}[H]
\caption{Earliest Arrival First Assignment of Patrons.}
\label{alg:earliest_arrival_first_assignment}
\begin{algorithmic}[1]
\Statex \textbf{Initialization:} 
\ForAll{$j \in \mathcal{J}$}
    \State $t_j \gets a_j$ \Comment{pointer to left endpoint of unassigned interval of request $j$.}
    \State $\paid(j) \gets 0$
    \State $\Lambda(j) \gets \del(j, [a_j, \freezetime(j)])$ \Comment{budget of request $j$.}
\EndFor
\ForAll{$(v, T)$}
  \State $\patrons(v, T) \gets \emptyset$
\EndFor

\Statex

\For{tree $T$ in $T_1, T_2, \ldots, T_k$}
  \State $Q\gets$ set of requests served in $T$
    \State $U \gets$ nodes $v\in T$ sorted first by increasing $\freezetime(v,T)$ and then in leaf-to-root order 
      \ForAll{$v \in U$ in sorted order}
        \State $J_v \gets \{j\in Q_v: \paid(j) < \Lambda(j) \text{ and } \freezetime(j)=\freezetime(v,T)\}$ \label{step:eaf-job-set} 
        \Statex \Comment{requests eligible to pay for $v$.}
        \State $\eta \gets w(v)$ \Comment{requirement of $v$.}
        \While{$\eta > 0$ and $J_v\ne \emptyset$} \label{line:eta_reaches_zero}
          \State $j \gets \argmin_{j \in J_v} \{a_j\}$ \label{line:find_j_from_Jv}
          \State $\lambda(j, v, T) \gets \min\{\Lambda(j) - \paid(j),\ \eta\}$ 
          \Statex \Comment{$j$ pays either its remaining budget or $v$'s remaining requirement.}
          \State Find $t' > t_j$ s.t. $\del(j, [t_j, t']) = \lambda(j, v, T)$
          
          \State $I(j, v, T) \gets [\,t_j,\ t'\,]$
          \State $\patrons(v, T) \gets \patrons(v, T) \cup\{ j\}$
          \State $\eta \gets \eta - \lambda(j,v,T)$
          \State $\paid(j) \gets \paid(j) + \lambda(j,v,T)$
          \State $t_j \gets t'$
          \State $J_v \gets J_v \setminus \{j\}$
        \EndWhile
      \EndFor
    \EndFor
\end{algorithmic}
\end{algorithm}

%% file: freeztime_tikz.tex

\begin{tikzpicture}[
    node distance=1.2cm and 1.2cm,
    base_circle/.style={
        circle, 
        draw=black, 
        thick, 
        minimum size=6mm, 
        inner sep=0pt,
        font=\sffamily\bfseries
    },
    base_square/.style={
        rectangle, 
        draw=gray,        
        fill=none,        
        dashed,           
        rounded corners=3pt, 
        thick, 
        minimum width=7mm, 
        minimum height=7mm, 
        inner sep=0pt,
        font=\small\sffamily
    },
    root/.style={base_circle, draw=cyan}, 
    w/.style={base_circle, draw=barred}, 
    u/.style={base_circle, draw=blue},
    v/.style={base_circle, draw=magenta},
    uc1/.style={base_circle, draw=orange},
    uc2/.style={base_circle, draw=green!60!black},
    root_small/.style={base_circle, draw=green!60!black, font=\small\sffamily, minimum size=4mm},  
    mid_small/.style={base_circle, draw=blue, font=\small\sffamily,         minimum size=4mm}, 
    leaf_small/.style={base_circle, draw=red, font=\small\sffamily,         minimum size=4mm}, 
    small_circle/.style={base_circle, draw=gray, font=\small\sffamily}, 
    edge/.style={-, thick},
    tiny_circle/.style={base_circle, draw=gray, font=\small\sffamily,  minimum size=4mm}, 
]

\definecolor{barred}{RGB}{220,80,80}
\definecolor{barblue}{RGB}{80,130,200}
\definecolor{bargreen}{RGB}{80,180,100}
\definecolor{barpurple}{RGB}{160,80,200}
\definecolor{baryellow}{RGB}{220,200,80}
\definecolor{barorange}{RGB}{230,140,50}

    \def\xr{9.5}
    \def\yr{9}

\node[
    ellipse, 
    draw=gray, 
    fill=gray!15, 
    dashed, 
    thick, 
    minimum width=1.5cm, 
    minimum height=2.75cm, 
    rotate=-45,  
    label={[text=gray!80!black]left:$T$} 
] at (\xr-0.5, \yr-0.5) {}; 

    \node[root] (r) at (\xr, \yr) {r};
    
    \node[w] (w) at (\xr-1,\yr-1) {$w$};

   
    \draw[edge] (r) -- (w);

    \node[u] (u) at (\xr-2,\yr-2) {u}; 
    \node[v] (v) at (\xr,\yr-2) {$v$};

    \node[uc1] (u1) at (\xr-2.5,\yr-3) {$u_1$};
    \node[uc2] (u2) at (\xr-1.5,\yr-3) {$u_2$};

    \draw[edge] (w) -- (u); 
    \draw[edge] (w) -- (v);
    \draw[edge] (u) -- (u1);
    \draw[edge] (u) -- (u2); 

    \node[base_square] (j2)  at (\xr-1.5,\yr-4)  {$j_2$};
    \node[base_square] (j1)  at (\xr-2.5,\yr-4)  {$j_1$}; 
    \node[base_square] (j3)  at (\xr,\yr-3)  {$j_3$};

    \draw[draw=gray, dashed, thick] (v.south) -- (j3.north);

    \draw[draw=gray, dashed, thick] (u1.south) -- (j1.north);
    \draw[draw=gray, dashed, thick] (u2.south) -- (j2.north);

        \definecolor{barred}{RGB}{220,80,80}
        \definecolor{barblue}{RGB}{80,130,200}
        \definecolor{bargreen}{RGB}{80,180,100}
        \definecolor{barpurple}{RGB}{160,80,200}
        \definecolor{baryellow}{RGB}{220,200,80}
        \definecolor{barorange}{RGB}{230,140,50}
        
        \def\off{-2}
        
        \node[left] at (0,4.85-\off)  {\tiny $j_1$};
        \node[left] at (0, 3.85-\off)  {\tiny $j_2$};
        \node[left] at (0, 3.1-\off)  {\tiny $j_3$};
        
        \draw[thick] (0.0, 4.85-\off) -- (4, 4.85-\off); 
        \draw[thick] (0.0, 3.85-\off) -- (4, 3.85-\off); 
        \draw[thick] (0.0, 3.1-\off) -- (4, 3.1-\off); 

        \fill[orange, draw=gray, line width=1pt]     (0.0, 7-\off) rectangle (.5, 5.05-\off);
        \fill[blue, draw=gray, line width=1pt]   (.51, 7-\off) rectangle (1, 5.05-\off);
        \fill[red, draw=gray, line width=1pt]   (1.01, 7-\off) rectangle (4, 5.05-\off);
        \fill[bargreen, draw=gray, line width=1pt]    (0, 4.55-\off) rectangle (2, 4.05-\off);
        \fill[red, draw=gray, line width=1pt]    (2.01, 4.55-\off) rectangle (4, 4.05-\off);
        \fill[magenta, draw=gray, line width=1pt]    (0, 3.5-\off) rectangle (4, 3.25-\off);

        \draw [thick, decorate, decoration={brace, amplitude=5pt}, draw=gray!80!black] 
    (-.1, 5.05-\off) -- (-0.1, 7.05-\off)   
    node [midway, left=21pt, text=gray!80!black] {$j_1$ (rate = 2)};

      \draw [thick, decorate, decoration={brace, amplitude=5pt}, draw=gray!80!black] 
    (-0.1, 4.05-\off) -- (-0.1, 4.55-\off) 
    node [midway, left=14pt, text=gray!80!black] {$j_2$ (rate = 0.5)};

      \draw [thick, decorate, decoration={brace, amplitude=5pt}, draw=gray!80!black] 
    (-0.1, 3.25-\off) -- (-0.1, 3.5-\off) 
    node [midway, left=10pt, text=gray!80!black] {$j_3$ (rate = 0.25)};

        \foreach \x/\t in {0,.5,1,2,3,4} {
            
            \draw[thick] (\x, 2.7-\off) -- (\x, 2.8-\off);
            
            \node[below] at (\x, 2.7-\off) {$\t$};
         }
           \draw[->, thick] (0,2.75-\off) -- (4.5,2.75-\off) node[right] {\textbf{Time}};
         \draw [thick, decorate, decoration={brace, amplitude=5pt}, draw=gray!80!black] 
    (0, 7.75-\off) -- (4, 7.75-\off) 
    node [midway, above=4pt, text=gray!80!black] {Freezetime in $T$};
         \draw[->, >=stealth, thick, text=gray!80!black, gray, dashed] (1, 7.4-\off) -- (1, 7-\off) node[above=7pt] {$u_1, u$};
      \draw[->, >=stealth, thick, text=gray!80!black, gray, dashed] (2, 7.4-\off) -- (2, 4.6 -\off) node[above=76pt] {$u_2$};
          \draw[->, >=stealth, thick, text=gray!80!black, gray, dashed] (4, 7.3-\off) to [bend left](4, 3.6 -\off) ;
          \node[gray] at (4, 7.5-\off){$v$};

\end{tikzpicture}

%% file: freezetime.tex
\subsection{Proof of the Freezetime Dominance Lemma (Lemma~\ref{lem:freezetime-dominance})}
\label{sec:freezetime_analysis}

We begin by stating a series of claims that establish consistency properties about the sets of requests serviced across different iterations of the algorithms as well as the reachtimes and freezetimes they determine. 

\begin{restatable}{lemma}{sattimemonotonicity}{\em [Monotonicity of Saturation Times]}
\label{lem:sattime-monotone}
Saturation times satisfy the following monotonicity properties:
\begin{itemize}
    \item $\sattime(\hat{Q}, u) \leq \sattime(Q, u)$ if $Q \subseteq \hat Q$.
    \item Let $u \in \tree$ and $Q$ be a set of requests such that $\del(Q, [\sattime(Q, u)]) = w(\tree_u^Q)$. Then for any $v \in \tree_u^Q$, $\del(Q_v, [\sattime(Q, u)]) \geq w(\tree_v^{Q_v})$.
    \item Let $u \in \tree$ and $Q$ be a set of requests such that $\del(Q, [\sattime(Q, u)]) = w(\tree_u^Q)$. Then for any $v \in \tree_u^Q$, 
    $\sattime(Q, v) \leq \sattime(Q, u)$.
\end{itemize}
\end{restatable}
\begin{proof}
    The first property is true by definition since any $Q' \subseteq Q_u$ that attains the minimum in $\sattime(Q, u) = \min\{ t \mid \exists Q' \subseteq Q_u \text{ such that } \del(Q', [t]) = w(\tree_u^{Q'})\}$ is also a subset of $\hat Q_u$.

    Let $t^* = \sattime(Q, u)$ for convenience. For the second property, suppose for contradiction that $\exists v \in \tree_u^Q$ such that $\del(Q_v, [t^*]) < w(\tree_v^{Q_v})$. Then consider $Q^* = Q \setminus Q_v$. Using the assumption that $\del(Q,[t^*]) = w(\tree_u^Q)$, we then have $\del(Q^*, [t^*]) = \del(Q, [t^*]) - \del(Q_v, [t^*]) > w(\tree_u^Q) - w(\tree_v^{Q_v}) = w(\tree_u^{Q^*}) \implies \sattime(Q^*, u) < t^*$. But, then we have $\sattime(Q, u) \leq \sattime(Q^*, u)$ leading to a contradiction.

    The third property follows from the definition of saturation times and the second property.
\end{proof}

\begin{restatable}{lemma}{allornone}
\label{lem:all-or-none}
Suppose $\Simulate(u, T)$ witnesses request set $Q$. Then either (i) none of the requests in $Q$ are served by $T$, or (ii) $\exists w \in T$ such that all requests in $Q$ are served during $\Simulate(w, T)$.
\end{restatable}

\begin{proof}
    We first observe that in node $u$'s simulation, either $\cst(T_{Q}) \leq b_u$, in which case all requests in $Q$ get served, or $\cst(T_{Q}) > b_u$, in which no request in $Q$ gets served. We next show by contradiction that if some request in $Q$ is served in another node's simulation, then all requests in $Q$ must be served in that node's simulation. Suppose that the first served request $j$ in $Q$ gets served due to node $w$'s simulation, but not all requests in $Q$ get served during $w$'s simulation. There are then two cases -- i.e., node $w$ can either be a descendant or ancestor of node $u$. 

    Suppose first that node $w$ is an ancestor of node $u$. Suppose further that $j'$ is the first request in $\mc{T}_u$ which is served due to node $w$'s simulation ($j'$ may either be in $Q$, or not). Then $j'$ must have been served because it was part of some set of request set $Q'$ that saturates some node $w'$ which is a child of $w$. Because not all requests in $Q$ were served, it must be that $Q \not \subseteq Q'$. But by definition, $Q$ is the earliest set of pending requests that is able to saturate some node $v$ which is a child of node $u$. In order to saturate node $w'$, some subset of requests in $\mc{T}_u$ must first saturate node $v$. Furthermore, because nodes are saturated in order from leaf to root and $w'$ is an ancestor of $v$, the set of pending requests which earliest saturates $w'$ necessarily includes the set of pending requests which earliest saturates $v$. Therefore, we must have $Q \subseteq Q'$. This implies that $w$ is not an ancestor of node $u$.

    Next suppose that node $w$ is a descendant of node $u$. But since we run simulations in a leaf-to-root order, we can only run a simulation from a descendant of $u$ if we previously serve some request $j'$ in $\mc{T}_u$ during a simulation from an ancestor $x$ of $u$. But we previously argued that then node $x$ must have served all of $Q$ during its simulation, which contradicts $j$ being served during node $w$'s simulation.
\end{proof}

\begin{restatable}{lemma}{earlydead}
\label{lem:early-deadline}
    Suppose $\Simulate(u, T)$ witnesses request set $Q$ and further suppose that $Q$ is served in $T$. Then for all $v \in T_u^{Q}$, $\freezetime(v, T) \leq \reachtime(Q, u)$.
\end{restatable}

\begin{proof}
     By Lemma $\ref{lem:all-or-none}$, we know that all requests in $Q$ are served at the same time, in the same node's simulation. Suppose that the sponsor set during the iteration that adds $Q$ to the tree is $Q'\supseteq Q$. Let node $v'$ be the highest node on $P(v, u)$ which is not in the tree when $Q$ is added.  Because $u$ is already in the tree when $Q$ is added, $v'$ is some descendant of $u$.  By the definition of reachtime, $Q$ saturates all vertices in the path from $v$ to $v'$ by $\reachtime(Q,u)$. Therefore, $Q'\supseteq Q$ also saturates all vertices in the path from $v$ to $v'$ by $\reachtime(Q,u)$. Then, by the definition of freezetime, $\freezetime(v,T)$ is at most $\reachtime(Q,u)$.
\end{proof}

\begin{restatable}{lemma}{boughtnodes}
\label{lem:bought-nodes-freeze-earlier}
    Suppose that node $u$ is sent once in tree $T$ at time $t$, and sent next in tree $T_{\nxt}$ at time $t_{\nxt}$. Further suppose that $u$ is overlapping in $T_{\nxt}$. Let $J$ be all requests pending at both times $t$ and $t_{\nxt}$. Then 
    \[
        \min_{Q \subseteq J} \reachtime(Q, u) \leq \freezetime(u, T_{\nxt}).
    \]
\end{restatable}

\begin{proof}
    To prove the lemma, it suffices to exhibit some $Q'\subseteq J$ such that $\reachtime(Q',u)\le \freezetime(u,T_{\nxt})$. Let $j$ be some request that caused $u$ to be overlapping, i.e. $j \in \patrons(u, T_{\nxt})$ and $a_j \leq t$. Let $v'$ be the unique child of $u$ on the path $P(v_j,u)$. Let $t'=\freezetime(u, T_{\nxt})$. We will build the set $Q'$ iteratively as follows. 

\begin{enumerate}
    \item Initialize $Q'=\{j\}$ and $S' = P(v_j,u')$. That is, $S'$ contains all vertices on the path from $v_j$ to $u'$. 
    \item In each iteration:
    \begin{enumerate}
        \item Add to $Q'$ all of the patrons in $T_{\nxt}$ of the vertices in $S'$. That is, $Q'=\{j\}\cup \cup_{v\in S'}\patrons(v,T_{\nxt})$.
        \item Add to $S'$ all of the vertices that lie on the path between some $w\in S'$ and some patron $j'$ of $w$. That is, $S'=\cup_{w\in S', j'\in \patrons(w,T_{\nxt})}P(v_{j'}, w)$.
    \end{enumerate} 
    \item We continue this process until no new vertices or patrons remain to be added. 
\end{enumerate}

    Observe that the above process must end in a finite number of steps. We now make the following observations. 

    First, recall that if $w\in T_{\nxt}$ and $j'\in\patrons(w,T_{\nxt})$, then $\freezetime(j')=\freezetime(w,T_{\nxt})$. Furthermore, by the definition of freezetime, for all $v\in P(v_j,w)$, we have $\freezetime(v,T_{\nxt})=\freezetime(w,T_{\nxt})$. Therefore, all requests added to $Q'$ and all vertices added to $S'$ by the above procedure have freezetime in $T_{\nxt}$ equal to $t'=\freezetime(u,T_{\nxt})$.

    Second, consider any patron $j'$ added to $Q'$ in the above procedure. $j'$ is either a patron of some vertex $w$ in $P(v_j,u')$ or a patron of some descendant of $w$. Since $j$ is also eligible to be a patron of $w$, then by the patron assignment rules, $j'$ arrived no later than $j$.    Therefore, we get that $Q'\subseteq J$.

    Finally, since $Q'$ contains all of the patrons of vertices in $S'$, by Lemma~\ref{lem:patrons_equal_cost}, the total delay accumulated by requests in $Q'$ until $t'$ is sufficient to pay for the cost of $S$. Therefore we get 
    \[\reachtime(Q',u)\le \sattime(Q',u')\le t'=\freezetime(u,T_{\nxt}),\]
    proving the lemma.
\end{proof}

Finally we are ready to prove the \emph{Freezetime Dominance Lemma} that we restate for convenience.

\originalstatefalse
\freezetimelem*

\begin{proof}
    Let $J$ be all requests pending at both times $t$ and $t_{\nxt}$. Suppose that $u$ invests in node $v$ in tree $T$ due to sponsor set $Q$. Then we have:
    \[
        \reachtime(Q, u) \leq \min_{Q' \subseteq J} \reachtime(Q', u) \leq \freezetime(u, T_{\nxt})
    \]
    where the second inequality is due to Lemma \ref{lem:bought-nodes-freeze-earlier}, and the first inequality follows from noting that any $Q'\subseteq J$ was pending at time $t$, but \Simulate($u, T$) witnessed $Q$ before other such sets.     
    Now we apply Lemma \ref{lem:all-or-none} to split the analysis into two cases: either all requests in $Q$ are served in tree $T$, or no request in $Q$ is served in tree $T$.  

    {\bf Case 1: $Q$ is served in $T$.} Since $v\in \tree_u^Q$, $v$ is also served in $T$. So we have:
    \[
        \freezetime(v, T^v) = \freezetime(v, T) \leq \reachtime(Q, u) \leq \freezetime(u, T_{\nxt})
    \]
    where the equality is because $T^v = T$ in this case and the first inequality is due to Lemma \ref{lem:early-deadline}. 

    {\bf Case 2: $Q$ is not served in $T$.} In this case, no requests in $Q$ can be served before $T_{\nxt}$, as otherwise $u$ must have been sent before $T_{\nxt}$ as well. Therefore, $Q\subseteq J$. Now consider the iteration during which $u$ was added to $T_{\nxt}$. If $u$ is part of the critical component in $T_{\nxt}$, then $\freezetime(u,T_{\nxt})=t_{\nxt}$ and $\reachtime(Q,u)\le t_{\nxt}$ implies that $Q$ (and, by extension, $v$) is also in the critical component of $T_{\nxt}$. Then $\freezetime(v,T^v) = \freezetime(v,T_{\nxt})=t_{\nxt}$ and the lemma follows.

    Suppose, instead, that $u$ is added to $T_{\nxt}$ during $\Simulate(u', T_{\nxt})$ and is sponsored by $Q'$. We claim that $Q\subseteq Q'$.      
    Assuming the claim, $v$ is added to $T_{\nxt}$ in the same simulation step as $u$, and hence $T^v=T_{\nxt}$. Let $P_{u,v}$ denote the path between nodes $u$ (exclusive) and $v$ (inclusive) in the tree. Then by definition of freezetime, we have
    \[
        \freezetime(v, T^v) = \freezetime(v, T_{\nxt}) = \max \{ \max_{x \in P_{u,v}} \sattime(Q'_x, x), \freezetime(u, T_{\nxt})\}.
    \]
    Let $x$ be an arbitrary node on $P_{u,v}$ and let $u_c$ be the child node of $u$ on the path $P_{u,v}$. By monotonicity of saturation times (Lemma \ref{lem:sattime-monotone}), we have
    \[\sattime(Q'_x, x) \leq \sattime(Q_x, x) \leq \sattime(Q_{u_c}, u_c) = \reachtime(Q, u).\]
    Substituting back into the freezetime expression above proves the lemma. 
    
    It remains to prove the claim that $Q\subseteq Q'$. By Lemma~\ref{lem:early-deadline}, $\reachtime(Q', u')\ge \freezetime(u,T_{\nxt})\ge \reachtime(Q,u)$. Let $t'=\reachtime(Q',u')$ and $\hat{Q}= Q'\cup Q$. Then it is easy to verify that for every vertex $v'\in\tree_{u'}^{\hat{Q}}$, we have $\sattime(\hat{Q}_{v'}, v') \leq t'$. This is because either $v'$ is in $\tree_{u'}^{Q'}$ and $Q'\subset \hat{Q}$ saturates $v'$ by $t'$, or $v'$ is in $\tree_u^Q$ and $Q\subset \hat{Q}$ saturates $v'$ by $\reachtime(Q,u)\le t'$. By the definition of reachtime, this implies $\reachtime(\hat{Q},u')\le t'$, contradicting the assumption that $\Simulate(u',T_{\nxt})$ witnesses $Q'$.
\end{proof}

%% file: deferred.tex
\section{Deferred Proofs (Delays)}
\label{sec:deferred-proofs-delays}

\subsection{Primal Analysis}
\label{sec:primal-deferred}
We will restate and prove Lemmas~\ref{lem:alg-bound-delays2} and \ref{lem:alg-bound-delays}. 

\algbounddelaystwo*

\begin{proof}
    Proof by contradiction. Suppose that $\delayalg > \sum_{i=1}^k w(T_i)$. Then there must exist some tree $T_i$ sent at time $t_i$ by Algorithm \ref{alg:multilevel_agg_delays} such that the delays accumulated by the requests served with $T_i$ is strictly greater than $w(T_i)$. Let $Q_i$ be the set of requests sent with tree $T_i$.  Because the delay function for any request $j$ is a continuous function that has value $0$ at time $a_j$, the total delay for all requests in $Q_i$ must also be a continuous function over time. This implies that there must be some time $t' < t_i$ and $Q' \subseteq Q_i$ such that the sum of delays accumulated by the requests in $Q'$ at time $t'$ is exactly equal to $w(T^{Q'})$. This implies that the request set $Q'$ is critical at time $t'$. Since Algorithm \ref{alg:multilevel_agg_delays} always serves all critical requests whenever they become critical, requests $Q'$ would have been served at time $t' < t_i$ which is a contradiction.
\end{proof}

\algbounddelays*

\begin{proof}
        We will construct a flow network and show a preflow where the total input is at most the RHS and the total output is the LHS.

\paragraph{Network Construction.}
    For each tree $T_i$ that is sent by the algorithm and for each node $v \in T_i$, add node $(v, T_i)$ to the network. Add a source node $a$ and a sink node $b$. Add an edge from source $a$ to each critical node, i.e., add edge $(a, (v, T_i))$ for each $v \in S_i$. Send a flow of $h(v) \cdot \hat{w}(v, t_i)$ through this edge, where $h(v)$ is the height of node $v$. Add an edge from each node sent by the algorithm to the sink, i.e., add edge $((v, T_i), b)$ for each $v \in T_i$. Send a flow of $w(v)$ through this edge.

    By this construction, the total flow leaving the source $a$ is at most the RHS and the total flow entering the sink $b$ is exactly the LHS. Therefore it suffices to add flow through the network in such a way that we have a feasible preflow. Suppose node $v$ invests amount $y(v, v', T_i)$ in node $v'$ in tree $T_i$. Let $T_j$ be the next tree where $v'$ is sent (note that potentially $T_j = T_i$). Then add an edge $((v, T_i), (v', T_j))$ and send flow $h(v') \cdot  y(v, v', T_i)$ through this edge.

\paragraph{Pre-flow Verification.} Consider any node $(v, T_i)$. If $v \in S_i$, then the inflow to this node is exactly $h(v) \cdot w(v)$, as the source sends $h(v) \cdot \hat{w}(v, t_i)$ flow and the other nodes send  $h(v) \cdot (w(v) - \hat{w}(v, t_i))$ flow due to the initial investment of $c_v(t_i)$ in node $v$ at time $t_i$. If $v \not\in S_i$, then the inflow is also exactly $h(v) \cdot w(v)$ due to flow from nodes investing in $v$.
The outflow from this node is at most 
\[
    w(v) + (h(v) - 1) \cdot w(v) = h(v) \cdot w(v).
\]
The first term is because each node sends $w(v)$ to the sink, the second term is because each node invests at most $w(v)$ in other nodes and only invests in nodes with heights strictly lesser than itself. Therefore, we have shown that the inflow is always at least as large as the outflow as desired.
\end{proof}

\subsection{EAF Assignment Analysis}\label{sec:eaf_proof}

\patronseqcost*
\begin{proof}
Properties (i)--(iii) follow immediately from construction. We now prove (iv).
Consider a tree $T$ serving the set $Q$ of requests and let $v\in T$. Let $t^*=\freezetime(v,T)$. To prove the lemma, it suffices to show the existence of a set 
$\tilde{J}_v\subset Q_v$ of requests such that $\freezetime(j)=t^*$ for all $j\in \tilde{J}_v$, and $\sum_{j\in \tilde{J}_v} \Lambda(j) \ge w(\tree_v^{\tilde{J}_v})$. This is because all the requests in $\tilde{J}_v$ participate in the iteration at Line~\ref{step:eaf-job-set} of Algorithm~\ref{alg:earliest_arrival_first_assignment} corresponding to freezetime $t^*$, and these requests only pay for vertices in $\tree_v^{\tilde{J}_v}$ prior to paying for $v$. Therefore, the claim shows that these requests have sufficient amount remaining to pay for $v$.

We will now prove the claim by exhibiting such a set $\tilde{J}_v$. Consider the iteration of Algorithm~\ref{alg:multilevel_agg_delays} where $v$ is added to the tree $T$. Let $Q'$ denote the sponsor set in this iteration, and note that $Q'\subseteq Q$. Let $Q_1 = Q'\cap \{j: \freezetime(j)=t^*\}$ and $Q_2 = Q'\setminus Q_1$. By definition, the freezetime of all requests in $Q_2$ is strictly larger than $t^*$. We claim that $\tilde{J}_v=Q_1$ suffices. To see this, we note that the requests in $Q'$ saturate all of the vertices in $\tree_v^{Q_1}$ by time $t^*$. However, the requests in $Q_2$ fail to saturate all vertices in $\tree_v^{Q'}\setminus \tree_v^{Q_1}$ by $t^*$. Therefore, requests in $Q_1$ alone suffice to saturate $\tree_v^{Q_1}$ by $t^*$, and so we are done.
\end{proof}

\subsection{Dual Analysis}

\delaysdualobjinc*

\begin{proof}
The claim follows by induction over the tree of recursive calls. Let $T_{\prev}$ be the last tree where $u$ was transmitted (strictly before $T$) and let $t_{\prev}$ be the time $T_{\prev}$ was sent. Our base case is when $u$ is a leaf node. In this case, it must be the case that for all $j$ that are patrons of $u$ in $T$, $a_j > t_{\prev}$, because otherwise request $j$ would have been sent with $T_{\prev}$. For every $j \in \patrons(u, T)$, we increase $\alpha_{j\tau}$ by $\delay(j, \tau) \cdot \tfrac{\gamma}{w(u)}$ for all $\tau \in I(j,u,T)$. The increase in the dual objective is therefore (using Lemma \ref{lem:patrons_equal_cost})
\begin{align*}
    \sum_{j \in \patrons(u, T)} \int_{I(j,u,T)} \alpha_{j\tau}d\tau &= \sum_{j \in \patrons(u, T)} \int_{I(j,u,T)} \delay(j, \tau) \cdot \frac{\gamma}{w(u)}d{\tau} \\
    &= \frac{\gamma}{w(u)} \cdot \sum_{j \in \patrons(u, T)} \lambda(j,u,T) \\
    &= \frac{\gamma}{w(u)} \cdot w(u) \\
    &= \gamma.
\end{align*}
Now suppose the inductive hypothesis is true for all $v \in A_u$ (recall that $A_u$ is the set of nodes that $u$ invests in during $T_{\prev}$). We will show that the inductive hypothesis holds for $u$ as well. First, if $a_j > t_{\prev}$ for all $j \in \patrons(u,T)$, then $u$ is disjoint and the same argument as above applies. Otherwise, if there exists a $j \in \patrons(u,T)$ where $a_j \leq t_{\prev}$, then we know that request $j$ was left pending during $\Explore(u)$ in tree $T_{\prev}$. But this implies that $u$ invested a total amount of exactly $w(u)$ in that round. So $\sum_{v \in A_u} y(u, v, T_{\prev}) = w(u)$. By the inductive hypothesis, the total increase in the dual objective is $\sum_{v \in A_u} y(u, v, T_{\prev}) \cdot \frac{\gamma}{w(u)} = \gamma$.
\end{proof}

\delayslimitedcharges*

\begin{proof}
Proof by induction on level of the tree starting from $r$. Our base cases are for $r, T$ for all $T$. Note that \Charge($r, T, \cdot$) can only be called once for any $T$, and further observe that $\gamma = w(r)$ by construction. This proves the base case. 

Now, consider some node $v$ and tree $T$ containing it, and let $t$ be the time it is sent. If $v$ is critical in $T$,  \Charge($v$, $T$, $\gamma$) is called once on Line~\ref{line:main-charge-call} with $\gamma=\hat{w}(v,t)$.

Next we will account for calls \Charge($v$, $T$, $\gamma$) by some ancestor $u$ of $v$. Let $t_0$ be the last time strictly before $t$ that $v$ was sent, or $0$ if no such time exists. Let $T_{u1}, T_{u2}, \cdots$ denote the trees containing $u$ that were sent strictly after $t_0$ but no later than $t$. Then, any call $\Charge(u, T_{ui}, \gamma'_{uij})$, for $i>1$, would have potentially generated a call $\Charge(v, T, \gamma_{uij})$, where $j$ indexes all calls made to $\Charge(u, T_{ui}, \cdot)$. Summing over all of the charges made by such calls, we get:
\begin{align*}
\sum_{u\in P(v)}\sum_{T_{ui}, i>1} \sum_j \gamma_{uij} & = \sum_{u\in P(v)}\sum_{T_{ui}, i>1} \sum_j y(u,v,T_{u (i-1)}) \cdot \frac{\gamma'_{uij}}{w(u)}\\
& \le \sum_{u\in P(v)}\sum_{T_{ui}}  y(u,v,T_{u (i-1)}) \cdot \frac{w(u)}{w(u)}\\
&\le c_v(t)
\end{align*}
where the first inequality used the induction hypothesis applied to calls $\Charge(u,T_{ui}, \cdot)$, and the second tallied the total amount invested in $v$ in the interval $(t_0, t]$ using the fact that $v$ was not sent in $(t_0, t)$.

Summing the two contributions implies that the total charge across all calls to $\Charge(v,T,\cdot)$ is at most $\hat{w}(v, t) + c_v(t) = w(v)$. 
\end{proof}

%% file: extra_deadlines.tex
\section{Deferred Proofs (Deadlines)}
\label{sec:extra_deadlines}

This section provides the proofs of all claims in Section~\ref{sec:deadlines}.

\dualobjinc*
\begin{proof}
    The lemma follows by induction over the tree of recursive calls. Our base case is when $u$ is a leaf node. In this case, it must be the case that $a_j > t_{\prev}$, because otherwise $v$ would have been sent with $T_{\prev}$. We therefore increase $\alpha_j$ by $\gamma$, which increases our dual objective by $\gamma$. Now suppose the inductive hypothesis is true for all $v \in A_u$. We will show that the inductive hypothesis holds for $u$ as well. First, if $a_j > t_{\prev}$, then we again increase $\alpha_j$ by $\gamma$, which increases the dual objective by $\gamma$. Whenever $a_j \leq t_{\prev}$, we know that request $j$ was left pending during $\Explore(u)$ in tree $T_{\prev}$. But this implies that $u$ invested a total amount of exactly $b_u = w(u)$ in that round. So $\sum_{v \in A_u} y(u, v, T_{\prev}) = w(u)$. By the inductive hypothesis, the total increase in the dual objective is $\sum_{v \in A_u} y(u, v, T_{\prev}) \cdot \frac{\gamma}{w(u)} = \gamma$.
\end{proof}

We will next prove Lemma~\ref{lem:deadlines-delta-alpha-3}. First, we state and prove two helpful claims that relate the deadline of a request charged for a vertex in tree $T$ to the time $T$ is scheduled.

\deadlineorderclaim*
\begin{proof}
    Let $T_{\prev}$ be the tree before $T$ when $u$ was previously sent, and suppose $T_{\prev}$ was sent at $t_{\prev}$ and $T'$ was sent at $t'$. Let $j_{\prev}$ be the critical request in $(T_{\prev})_v$. As defined in the algorithm, $T'$ is the next tree after $T_{\prev}$ (inclusive of $T_{\prev}$) where node $v$ is sent. Because node $v$ is sent when request $j_{\prev}$ is served, this implies that $j_{\prev}$ must still have be pending at time $t'$. Because $j'$ is the critical request of $T'_v$, this implies that $d_{j'} \leq d_{j_{\prev}}$.

    We will argue that $d_{j_{\prev}} \le d_j$, which suffices to prove the claim. Suppose for contradiction that $d_{j} < d_{j_{\prev}}$. If $j \in \tree_v$, then $d_{j_{\prev}} \leq d_j$ by definition of critical request, so it must be that $j \notin \tree_v$. If $d_{j} < d_{j_{\prev}}$, then every node on the path of $j$ must have been bought before $u$ started investing in $v$, which means that $j$ must have been sent with tree $T_{\prev}$. However, $j$ is still pending until tree $T$ is sent, which is a contradiction.
 \end{proof}

\begin{claim}\label{claim:ddl_v_less_than_ddl_r}
Suppose $\Charge(u, T, \gamma)$ is a top-level recursive call in Line \ref{line:top_level_recursive} of Algorithm \ref{alg:deadlines-analysis} that results in an increase in $\alpha_j$ for some request $j = (x, a, d)$. Then $d_j \leq t$ where $t$ is the time when tree $T$ was transmitted. 
\end{claim}
\begin{proof}
 This follows by repeated calls of Claim \ref{claim:ddl_v_less_than_ddl_u}. Consider the list of \Charge \ calls which results in an increase of $\alpha_j$, and order the calls from leaf node $x$ to $u$. Suppose the \Charge \ call $k$ takes as arguments $v_k$ and $T_k$. Let $d_k$ be the deadline of the critical request of the subtree of $T_k$ rooted at $v_k$. Then by Claim \ref{claim:ddl_v_less_than_ddl_u}, $d_{k-1} \leq d_k$, which by implies that $d_j = d_1 \leq d_u = t$.
\end{proof}

\deadlinesdeltaalpha*

\begin{proofof}{Lemma~\ref{lem:deadlines-delta-alpha-3} (a)}
    Fix a request $j$. Let $t'$ be the time of the earliest root call that caused $\alpha_j$ to be incremented. We will prove the claim for all $t\in [a_j, d_j]$ by induction over the course of the algorithm. At the beginning of the algorithm, all $\alpha$s and $\delta$s are zero, so the claim holds.
    
    Consider an iteration where the algorithm increases $\alpha_{j}$. Let $\Charge(u_1, T_1, \gamma_1)$ be the last call to $\Charge$ that makes this increase. Observe that $u_1$ is disjoint in $T_1$. We will now consider the sequence of nested charge calls that lead to the call $\Charge(u_1, T_1, \gamma_1)$. Let this sequence be of length $k$. The $i$th call in this sequence is denoted $\Charge(u_i, T_i, \gamma_i)$, and for $i>1$, $\Charge(u_i, T_i, \gamma_i)$ calls $\Charge(u_{i-1}, T_{i-1}, \gamma_{i-1})$. 
    
    Let $T_{\prev, i}$ denote the last tree strictly before $T_i$ where $u_i$ was transmitted. Let $t_i$ denote the time at which $T_i$ is sent, and $t_{\prev, i}$ the time at which $T_{\prev, i}$ is sent. 
    Observe that $\Charge(u_1, T_1, \gamma_1)$ increases $\alpha_{j}$ by $\gamma_j$. In the sequence of $\Charge$ calls defined above, for each $u_i$, $i\in [k]$, we increase $\delta_{u_i\tau j}$ 
    by the same $\gamma_j$ amount over the interval $(t_{\prev, i}, t_i]$. In order to prove the lemma, we must show that these intervals together cover the entire interval $[a_j, d_j]$.
    
    In order to see this, we first observe that since $u_1$ is disjoint in $T_1$, $a_j>t_{\prev, 1}$. Furthermore, for $i>1$, $t_{\prev, i}\le t_{i-1}$ by the definition of $T_{i-1}$. Therefore, all of these intervals together cover $[a_j, t_k]$. Finally, Claim \ref{claim:ddl_v_less_than_ddl_r} exactly says that $d_j \le t_k$, implying that all of these intervals together cover $[a_j, d_j]$.
\end{proofof}

\begin{proofof}{Lemma~\ref{lem:deadlines-delta-alpha-3} (b)}
    Fix a $v$ and $t$. Define $t_1$ as the latest time strictly before $t$ when $v$ was transmitted and define $t_2$ as the earliest time $\geq t$ when $v$ was transmitted (and let $t_2 = \infty$ if $v$ is not transmitted again after $t$). Observe that $\delta_{vt}$ is increased by $\gamma$ exactly whenever \Charge($v$, $T_2$, $\gamma$) is called for the tree $T_2$ sent at time $t_2$. By Claim~\ref{claim:limited_charges-deadlines}, the sum of these increases is at most $w(v)$ as desired.
\end{proofof}

Claim \ref{claim:limited_charges-deadlines} is helpful for proving Lemma \ref{lem:deadlines-delta-alpha-3} (b). 

\begin{claim}\label{claim:limited_charges-deadlines}
Fix a $v$ and a $T$. The sum of all $\gamma$ for calls $\Charge(v,T, \gamma)$ is at most $w(v)$.
\end{claim}
\begin{proof}
This proof follows exactly as in Claim \ref{claim:delays_limited_charges}.
\end{proof}

%% file: tight_example.tex
\section{Tight Example for Algorithm \ref{alg:multilevel_agg}}\label{sec:tight_example}

In this section, we give an example where the cost incurred by Algorithm \ref{alg:multilevel_agg} is $D$ times the optimal cost. 

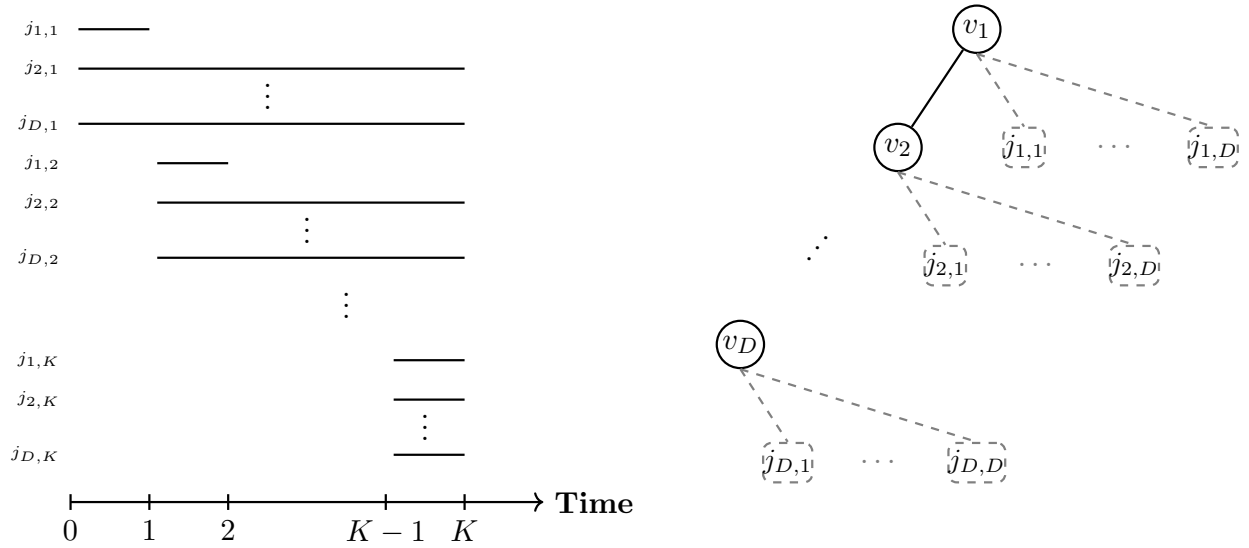
\begin{figure}[h]
    \centering
    \begin{adjustbox}{width=1\textwidth}
    \input{tight_deadlines_tikz}
    \end{adjustbox}
    \caption{A tight instance for Algorithm \ref{alg:multilevel_agg}. The tree $\mathcal{T}$ consists of $D$ levels where each node has a unit cost. Horizontal segments represent requests $j_{i,k}$ arriving at vertex $v_i$, with their horizontal span indicating the interval between arrival and deadline. Specifically, a request $j_{i,k}$ arrives at time $k-1+\epsilon$ for some small $\epsilon > 0$. For requests arriving at root ($i=1$), the deadline of $j_{1,k}$ is $k$; for all other request ($i > 1$), the deadline is $K$.}
    \label{fig:tight_deadlines_example}
\end{figure}

Consider the tree $\mathcal{T}$ and the request sequence $\mathcal{J}$ depicted in Figure \ref{fig:tight_deadlines_example}. In this instance, $\mathcal{T}$ is a path of depth $D$ rooted at vertex $v_1$ ($r=v_1$), where each level $i$ contains an internal node $v_i$ with cost $w(v_i) = 1$. At each level, $K$ requests arrive at the node, where $K$ is a large constant. We index these requests first by their arrival node and second by their arrival order relative to other requests at the same node. Specifically, for every integer $k \in [1, K]$, request $j_{i,k}$ arrives at node $v_i$ at time $k - 1 + \epsilon$. Regarding deadlines, all requests have a deadline of $K$, with the exception of those arriving at node $v_1$, where each request $j_{1,k}$ has a deadline of $k$.

The deadlines are structured to compel premature aggregation. In this instance, Algorithm \ref{alg:multilevel_agg} transmits $D$ separate trees—one at each integer time step from $t = 1$ to $t = K$. At each such $t$, the request at the root $v_1$ becomes critical, providing the root with a budget of $w(v_1) = 1$. According to the algorithm’s logic, this budget propagates down the path: each internal node $v_i$ facilitates the purchase of itself and the subsequent node $v_{i+1}$. Consequently, at every integer time step, the algorithm pays for the entire path of depth $D$, resulting in a total cost of $D \cdot K$ over the full sequence.

In contrast, an optimal offline strategy would satisfy only the critical root requests at each integer time step $t \in [1, K-1]$, deferring all other requests until the final time step $t = K$. At $t = K$, a single transmission of the full path satisfies all remaining pending requests at the deeper levels of the tree. Consequently, the total cost of this optimal offline solution is $(K - 1) + D$.

As $K \to \infty$, the ratio of the algorithm's cost to the optimal cost behaves as follows:$$\frac{\text{cost}(\text{ALG})}{\text{cost}(\text{OPT})} = \frac{D \cdot K}{(K - 1) + D} \longrightarrow D.$$ The fundamental issue highlighted by this example is that all of the extra cost that Algorithm \ref{alg:multilevel_agg} incurs in each tree ends up being totally unnecessary, because it is possible to send every internal node except for $v_1= r$ only once, at the very end.

%% file: tight_deadlines_tikz.tex


\begin{tikzpicture}[
    node distance=1.2cm and 1.2cm,
    base_circle/.style={
        circle, 
        draw=black, 
        thick, 
        minimum size=6mm, 
        inner sep=0pt,
        font=\sffamily\bfseries
    },
    base_square/.style={
        rectangle, 
        draw=gray,        
        fill=none,        
        dashed,           
        rounded corners=3pt, 
        thick, 
        minimum width=5mm, 
        minimum height=5mm, 
        inner sep=0pt,
        font=\small\sffamily
    },
    tree_node/.style={base_circle, draw=black},
    edge/.style={-, thick} 
]

\definecolor{barred}{RGB}{220,80,80}
\definecolor{barblue}{RGB}{80,130,200}
\definecolor{bargreen}{RGB}{80,180,100}
\definecolor{barpurple}{RGB}{160,80,200}
\definecolor{baryellow}{RGB}{220,200,80}
\definecolor{barorange}{RGB}{230,140,50}

    \def\xr{11.5}
    \def\yr{9}

    \node[tree_node] (v1) at (\xr, \yr) {$v_1$};

    \node[base_square] (j11)  at (\xr+.6,\yr-1.5)  {$j_{1,1}$};
    \node[base_square] (j1D)  at (\xr+3,\yr-1.5)  {$j_{1,D}$};
    \draw[draw=gray, dashed, thick] (v1.south) -- (j11.north);
    \draw[draw=gray, dashed, thick] (v1.south) -- (j1D.north);

    \path (j11) -- (j1D) node [midway, gray] {$\cdots$};
    
    \node[tree_node] (v2) at (\xr-1, \yr-1.5) {$v_2$};
    \draw[edge] (v1) -- (v2);

    \node[base_square] (j21)  at (\xr-.4,\yr-3)  {$j_{2,1}$};
    \node[base_square] (j2D)  at (\xr+2,\yr-3)  {$j_{2,D}$};
    \draw[draw=gray, dashed, thick] (v2.south) -- (j21.north);
    \draw[draw=gray, dashed, thick] (v2.south) -- (j2D.north);
    
    \path (j21) -- (j2D) node [midway, gray] {$\cdots$};

    \node[tree_node] (vD) at (\xr-3, \yr-4) {$v_D$};
    \path (v2) -- (vD) node [midway, sloped] {$\dots$};
    
    \node[base_square] (jD1)  at (\xr-2.4,\yr-5.5)  {$j_{D,1}$};
    \node[base_square] (jDD)  at (\xr,\yr-5.5)  {$j_{D,D}$};
    \draw[draw=gray, dashed, thick] (vD.south) -- (jD1.north);
    \draw[draw=gray, dashed, thick] (vD.south) -- (jDD.north);

    \path (jD1) -- (jDD) node [midway, gray] {$\cdots$};
    
    \def\off{-4}
    
    \node[left] at (0,5-\off)  {\tiny $j_{1,1}$};
    \node[left] at (0,4.5-\off)  {\tiny $j_{2,1}$};
    \node[left] at (0,3.8-\off)  {\tiny $j_{D,1}$};
    
    \draw[thick] (0.1, 5-\off) -- (1, 5-\off); 
    \draw[thick] (0.1, 4.5-\off) -- (5, 4.5-\off); 
    \draw[thick] (0.1, 3.8-\off) -- (5, 3.8-\off);
    
    \node at (2.5, 4.25-\off) {$\vdots$}; 

    \node[left] at (0,3.3-\off)  {\tiny $j_{1,2}$};
    \node[left] at (0,2.8-\off)  {\tiny $j_{2,2}$};
    \node[left] at (0,2.1-\off)  {\tiny $j_{D,2}$};
    
    \draw[thick] (1.1, 3.3-\off) -- (2, 3.3-\off); 
    \draw[thick] (1.1, 2.8-\off) -- (5, 2.8-\off); 
    \draw[thick] (1.1, 2.1-\off) -- (5, 2.1-\off); 
    
    \node at (3, 2.55-\off) {$\vdots$}; 

    \node at (3.5, 1.6-\off) {$\vdots$}; 
    
    \node[left] at (0,.8-\off)  {\tiny $j_{1,K}$};
    \node[left] at (0,.3-\off)  {\tiny $j_{2,K}$};
    \node[left] at (0,-.4-\off)  {\tiny $j_{D,K}$};
    
    \draw[thick] (4.1, .8-\off) -- (5, .8-\off); 
    \draw[thick] (4.1, .3-\off) -- (5, .3-\off); 
    \draw[thick] (4.1, -.4-\off) -- (5, -.4-\off); 

    \node at (4.5, .05-\off) {$\vdots$}; 
    
    \def\off{-1}

    \draw[->, thick] (0,2-\off) -- (6,2-\off) node[right] {\textbf{Time}};
    \foreach \x in {0, 1, 2, 5} {
        
        \draw[thick] (\x, 1.9-\off) -- (\x, 2.1-\off) ;

        \ifnum\x<5
            \node[below] at (\x, 1.9-\off) {$\x$};
        \fi
        \ifnum\x=5
            \draw[thick] (\x-1, 1.9-\off) -- (\x-1, 2.1-\off) ;
            \node[below] at (\x-1, 1.9-\off) {$K-1$};
            \node[below] at (\x, 1.9-\off) {$K$};
        \fi
        
    }
         
\end{tikzpicture}

%% file: related.tex
\section{Related Work}
\label{sec:related}

\mlap~captures settings where one must balance the benefits of delaying tasks against the aggregated costs, whether in the form of hard deadlines or accumulated delay costs. It generalizes several well-studied models, including the TCP Acknowledgment Problem (for depth-1 trees)\cite{buchbinder2007online,dooly1998tcp, karlin2001dynamic}, the Joint Replenishment Problem (\jrp) (for depth-2 trees) \cite{bienkowski2014better, brito2012competitive, buchbinder2013online}, and multi-level message aggregation (for trees of arbitrary depth) \cite{bienkowski2020online, azar2019general}.

\mlapdelay~and its special case \mlapd~were originally formulated by \cite{bienkowski2016online}, who proposed $O(D^4 2^D)$- and $O(D^2 2^D)$-competitive algorithms respectively for these problems, where $D$ denotes the depth of the underlying tree. For \mlapd, \cite{buchbinder2017depth} developed an $O(D)$-competitive deterministic algorithm, significantly improving over the original exponential bounds. \cite{buchbinder2017depth} show that it suffices to assume that the tree is hierarchically separated (HST) and devise a clever combinatorial argument to show that for each tree transmitted by the algorithm, an optimal solution incurs a marginal cost of at least $\Omega(\frac{1}{D})$ fraction to serve the same set of requests.
More recently, \cite{mcmahan2021d} gave an improved algorithm with a competitive ratio of exactly $D$.

For \mlapdelay\ with general delay functions, the algorithm of \cite{buchbinder2017depth} was extended by \cite{talmonthesis} to obtain a $O(D^3)$-competitive algorithm. \cite{azar2019general} devised a new combinatorial framework based of preflow constructions for metric optimization problems with delay which gave an improved $O(D^2)$-competitive algorithm. The best known lower bound for \mlap\ is $4$ \cite{bienkowski2020online} which holds on paths of unbounded depth even for the special cases of deadlines and linear delay functions. 

Recent research has explored further generalizations of \mlap\ beyond the standard clairvoyant adversarial model. \cite{ezra2024universal} studied the non-clairvoyant setting, where delay functions are unknown prior to request arrival. They developed a modular framework that achieves tight $O(\sqrt{n})$-competitive algorithms for \mlap~and related subadditive joint replenishment problems. Simultaneously, \cite{mari2024online} introduced a stochastic version of \mlap~with Poisson arrivals and obtained a constant ratio of expectations, demonstrating that probabilistic information on request arrivals can significantly improve performance guarantees.  Recent research \cite{moseley2025putting,azar2026online,shmoys2026improved} has also explored  generalizations of JRP where requests have a target time when they wish to be served and incur both a holding cost if served too early and a backlog cost if served too late.

In the \textit{offline setting}, where all requests are known a priori, \mlap~(both variants) is known to be APX-hard even for depth-2 trees~\cite{arkin1989computational, bienkowski2015approximation, nonner20095}. For \mlapd, there exists a $2$-approximation by \cite{becchetti2009latency, bienkowski2020online}. For general \mlapdelay, \cite{pedrosa2018approximation} develop a 
$(2+\epsilon)$-approximation algorithm adapting techniques from the multi-stage assembly problem \cite{levi2004primal}.

The framework underlying \mlap\ extends to a wider class of online optimization problems with temporal constraints. \cite{azar2019general} studied the \textit{Online Service with Delay (OSD)} (introduced by \cite{azar2017online}), in which a single server moves through metric space to serve requests that accumulate delay. They also studied \textit{Online Facility Location with Deadlines/Delay} where facilities can be opened temporarily at a fixed cost. Pending requests (with deadline or incurring delays) may connect to open facilities to receive service. 
For both problems, they obtained $O(\log^2 n)$-competitive algorithms, where $n$ is the number of points in the metric space. Later $O(\log n)$-competitive deterministic algorithms were obtained for both these problems by \cite{touitou2023improved} and \cite{azar2020beyond} respectively.